\documentclass{article}
\usepackage[margin=1in]{geometry}
\usepackage{amsmath}
\usepackage{amssymb}
\usepackage{graphicx}
\usepackage{parskip}
\usepackage{natbib}
\usepackage[hidelinks]{hyperref}
\usepackage{booktabs}
\usepackage{color}
\usepackage{comment}
\usepackage{subcaption}
\usepackage{enumerate}
\emergencystretch=\maxdimen
\usepackage{amsthm}

\newtheorem{proposition}{Proposition}
\newtheorem{corollary}{Corollary}
\title{On modelling bicycle power-meter measurements}
\author{%
	Tomasz Danek\footnote{%
		AGH\,---\,University of Science and Technology, Krak\'ow, Poland: {\tt tdanek@agh.edu.pl}}\,,
	Michael A. Slawinski\footnote{%
		Memorial University of Newfoundland, Canada: {\tt mslawins@mac.com}; {\tt theodore.stanoev@gmail.com}}\,,%
	\addtocounter{footnote}{-1}
	Theodore Stanoev\footnotemark
}
\date{}
\begin{document}
\maketitle
\begin{abstract}
We combine power-meter measurements with GPS measurements to study the model that accounts for the use of power by a cyclist.
The model takes into account the change in elevation and speed along with adverse effects of air, rolling and drivetrain resistance.
The focus is on estimating the resistance coefficients using numerical optimization techniques to maintain an agreement between modelled and measured power-meter values, which accounts for the associated uncertainties.
The estimation of coefficients is performed for two typical scenarios of road cycling under windless conditions, along a course that is mainly flat as well as a course of near constant inclination.
Also, we discuss relations between different combinations of two model parameters, where other quantities are constant, by the implicit function theorem.
Using the obtained estimates of resistance coefficients for the two courses, we use the mathematical relations to make inferences on the model and physical conditions.
Along with a discussion of results, we provide two appendices.
In the first appendix, we illustrate the importance of instantaneous cadence measurements.
In the second, we consider the model in constrained optimization using Lagrange multipliers.
\end{abstract}
\section{Introduction}
\label{sec:Introduction}
Using power meters to study cycling performance combines two distinct realms absent of an explicit relation: mechanical measurements and physical conditions.
A hypothetical relation between them is offered by mathematical modelling.
For instance from power metres alone, it is not possible to determine whether or not a high power measurement that results in a low ground speed is the consequence of a strong headwind, a steep climb, their combination, or a completely different factor, such as an unpaved road.
Any determination can only be postulated, under various assumptions, as a model, without a claim to uniqueness of solution.

Many studies examine the physics of cycling.
For instance, there are wind-tunnel studies to measure the aerodynamics of bicycle wheels~\citep[e.g.,][]{GreenwellEtAl1995}, and the studies to estimate the accuracy of power measurements based on the frequency of the pedal-speed measurements~\citep[e.g.,][]{Favero2018}.
There are studies to examine power required to overcome the effect of winds, taking into account tire pressure, wheel radius, altitude and relative humidity~\citep[e.g.,][]{OldsEtAl1995}, as well as the aerodynamic drag, rolling resistance, and friction in the drivetrain~\citep[e.g.,][]{MartinEtAl1998}.
There are studies to estimate model parameters from measurements on the road~\citep[e.g.,][]{Chung2012} and to devise optimal speeds for time trials on closed circuits in the presence of wind~\citep[e.g.,][]{Anton2013}.
There are studies to investigate the aerodynamics of track cycling to predict the individual-pursuits times~\citep[e.g.,][]{UnderwoodJermy2014} and to simulate cyclist slip and steer angles necessary to navigate turns on a banked track~\citep[e.g.,][]{FittonSymons2018}.
There are graduate theses in mechanical engineering~\citep[e.g.,][]{Moore2008,Underwood2012}.
It is evident that the science of cycling is a rich field that combines theoretical, computational and experimental aspects of such disciplines as mathematical physics, fluid mechanics, as well as the optimization and approximation theories.

We begin this article by formulating a mathematical model that accounts for the power generated by a cyclist to propel the motion of a bicycle-cyclist system, in terms of the forces that oppose its motion.
Next, we present power-meter and GPS measurements for a flat and inclined course in Northwestern Italy: the former is between Rivalta Bormida and Pontechino, in Piedmont; the latter is between Rossiglione and Tiglieto, in Liguria. 
We use the model and data to estimate parameter values of air, rolling and drivetrain resistance by seeking an acceptable agreement between obtained measurements and model retrodictions, as well as examining their uncertainties.
Then, by invoking the implicit function theorem, we derive explicit expressions for the rates of change between the model parameters and use their values to make model and physical inferences.
We proceed by considering various aspects of the model, such as the one-to-one relationship between power and speed, derivations of terms used in the model and the effects of wheel rotation.
We conclude by a discussion of results.

This article contains two appendices.
In the first appendix, we discuss the discrepancies of power-meter calculations based on instantaneous cadence as opposed to its revolution average.
In the second, we perform optimizations of ground speed based on physical constraints using the method of Lagrange multipliers, which are familiar to mathematical physicists, but might be less so to a broad range of sport scientists.
The same is true of Sections~\ref{sec:AirDens} and \ref{sec:AirResCoeff}.
We include them as auxiliary material to enhance and facilitate the understanding of the presented material. 
\section{Formulation}
\subsection{Mathematical model}
We consider a mathematical model of the power required to overcome the forces, $F_{\!\leftarrow}$, that oppose the translational motion, for ground speed $V_{\!\rightarrow}$, of a bicycle-cyclist system.
Herein,
\begin{align}
	\label{eq:model}
	P
	&=
	F_{\!\leftarrow}\,V_{\!\rightarrow}
	\\
	\nonumber 
	&=
	\quad\frac{
		\overbrace{\vphantom{\left(V\right)^2}\,m\,g\sin\theta\,}^\text{change in elevation}
		+
		\!\!\!\overbrace{\vphantom{\left(V\right)^2}\quad m\,a\quad}^\text{change in speed}
		+
		\overbrace{\vphantom{\left(V\right)^2}
			{\rm C_{rr}}\!\!\!\underbrace{\,m\,g\cos\theta}_\text{normal force}
		}^\text{rolling resistance}
		+
		\overbrace{\vphantom{\left(V\right)^2}
			\,\tfrac{1}{2}\,{\rm C_{d}A}\,\rho\,
			(
				\!\!\underbrace{
					V_{\!\rightarrow}+w_{\leftarrow}
				}_\text{air-flow speed}\!\!
			)^{2}\,
		}^\text{air resistance}
	}{
		\underbrace{\quad1-\lambda\quad}_\text{drivetrain efficiency}
	}\,V_{\!\rightarrow}\,,
\end{align}
where $F_{\!\leftarrow}$ consists of the following quantities\footnote{%
For consistency with power meters, whose measurements are expressed in watts, which are $\rm{kg\,m^2/s^3}$\,, we use the {\it SI} units for all quantities.
Mass is given in kilograms,~$\rm{kg}$\,, length in meters,~$\rm{m}$\,, and time in seconds,~$\rm{s}$\,; hence, speed is in $\rm{m/s}$\,, change in speed in $\rm{m/s^2}$\,, and force in newtons,~$\rm{kg\,m/s^2}$\,; angles are in radians.
Also, for the convenience of the reader, we round all presented values to four significant figures.}:
\begin{itemize}
	\item 
	$m$\,: combined mass of the bicycle and cyclist, whose units are kg;
	\item 
	$g$\,: acceleration due to gravity, whose effects are illustrated in Figure~\ref{fig:FigNewton} and units are ${\rm m/s}^2$;
	\item 
	$\theta$\,: slope of an incline, whose units are rad;
	\item 
	$a$\,: temporal rate of change of ground speed, $V_{\!\rightarrow}$, whose acceleration units are ${\rm m/s}^2$;
	\item 
	$\rho$\,: air density, formulated in Section~\ref{sec:AirDens}, is
	\begin{equation}
		\label{eq:DenAlt}
		\rho = 1.225\exp(-0.0001186\,h)\,,
	\end{equation}
	whose units are ${\rm kg}/{\rm m}^3$, and where $h$ is the altitude, in metres, above the sea level;
	\item 
	$w_{\leftarrow}$\,: wind component opposing the motion, whose units are m/s;
	\item 
	$\rm C_{rr}$\,: unitless rolling-resistance coefficient; in a manner analogous to the friction on the plane inclined by~$\theta$, $\rm C_{rr}$ is a proportionality constant between the maximum force,~$mg$\,, and the force normal to the surface,~$mg\cos\theta$;
	\item 
	$\rm C_{d}A$\,: air-resistance coefficient, whose formulated is discussed in detail in Section~\ref{sec:AirResCoeff}, which is the product of a unitless drag coefficient,~$\rm C_{d}$\,, and a frontal surface area,~$\rm A$, whose units are $\rm m^2$;
	\item 
	$\lambda$\,: unitless drivetrain-resistance coefficient to account for the loss of power between the power meters and the propelling rear wheel; if power meters are in the pedals, $\lambda$ includes the resistance of bottom bracket, chain, rear sprocket and rear-wheel hub; it also includes losses due to the flexing of the frame; if power meters are in the rear-wheel hub,~${\lambda\approx 0}$.
\end{itemize}
In keeping with the translational motion of the bicycle-cyclist system, we omit the effects due to the rotation of wheels in model~\eqref{eq:model}.
Implicitly, these effects are included in coefficients $\rm C_dA$, $\rm C_{rr}$ and $\lambda$.
However, we discuss how to accommodate explicitly their effects on air resistance in Section~\ref{sec:AirResRot} and on change in speed in Section~\ref{sec:RotEff}.
Furthermore, to include air density as an input parameter, we do not explicitly invoke humidity or atmospheric pressure.
The four summands in the numerator of model~\eqref{eq:model} are forces to account for
\begin{itemize}
	\item
	change in elevation: increases the required power if $\theta>0$\,, decreases if $\theta<0$ and has no effect if $\theta=0$\,; it is associated with the change in potential energy,
	\item
	change in speed: increases the required power if $a>0$\,, decreases if $a<0$ and has no effect if $a=0$\,; it is associated with the change in kinetic energy, which is not lost unless the rider brakes,
	\item
	rolling resistance: increases the required power,
	\item 
	air resistance: increases the required power if the speed of the air flow relative to the cyclist is positive, $(V_{\!\rightarrow}+w_{\leftarrow})>0$\,, decreases if $(V_{\!\rightarrow}+w_{\leftarrow})<0$ and has no effect if $(V_{\!\rightarrow}+w_{\leftarrow})=0$.
\end{itemize}
Similar models\,---\,exhibiting a satisfactory empirical adequacy\,---\,are used in other studies~\citep[e.g.,][]{MartinEtAl1998}.
\begin{figure}
\centering
\includegraphics[scale=0.4]{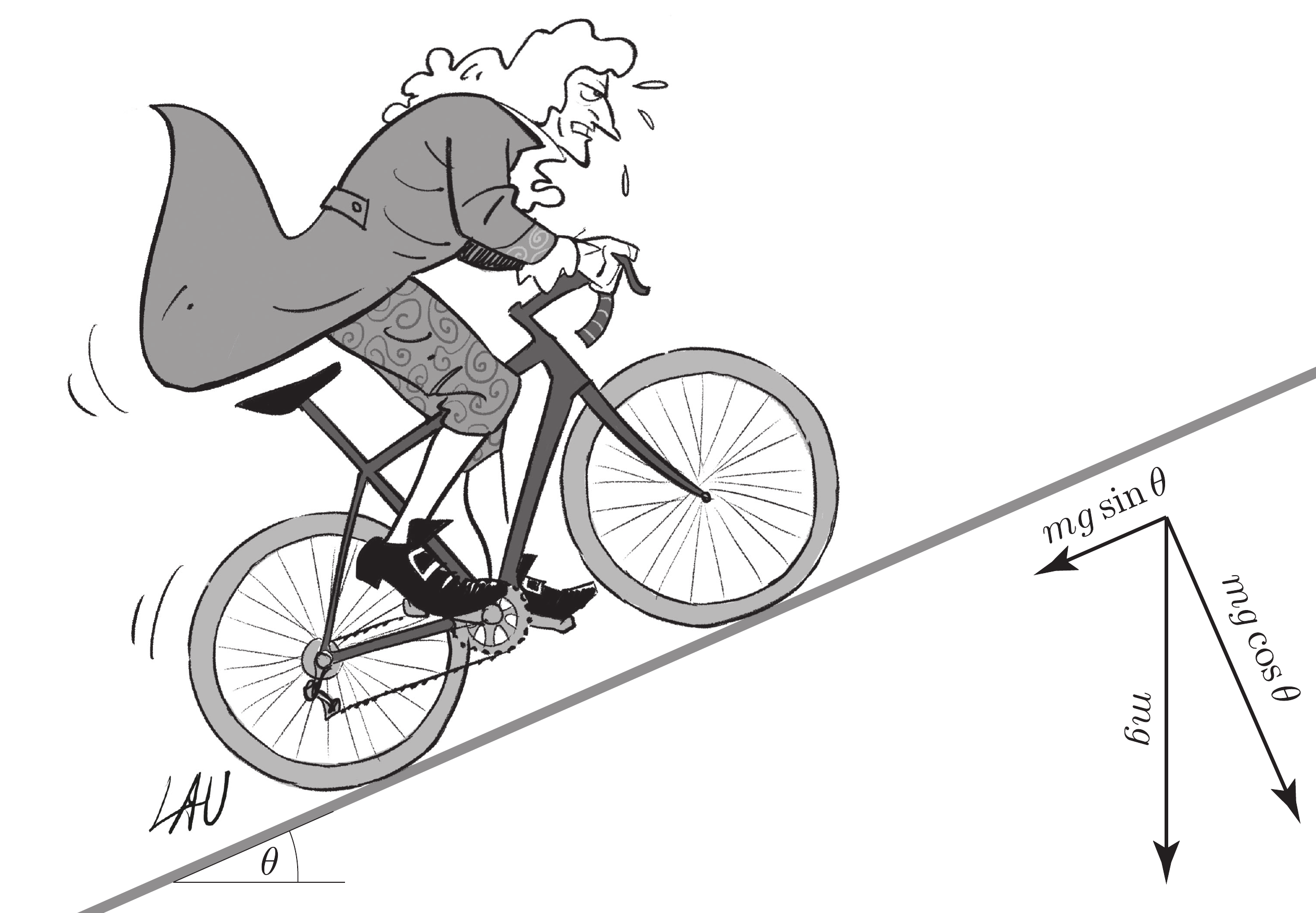}
\caption{\small Isaac Newton subject to the effects of gravity in model~\eqref{eq:model}.}
\label{fig:FigNewton}
\end{figure}

\subsection{Physical measurements}
\label{sec:PhysMeas}
Power is the rate at which work is done; hence, it is equal to the amount of work divided by the time it takes to do it, which is tantamount to the product of force and speed,
\begin{equation}
	\label{eq:formula}
	p
	=
	f_{\circlearrowright} v_{\circlearrowright}
	\,.
\end{equation}
In the context of cycling, $f_{\circlearrowright}$ is the force applied to pedals and $v_{\circlearrowright}$ is the speed with which the rotating pedals cover the distance along the circumference of their rotation, which means that $v_{\circlearrowright}$ is proportional to the length of the crank.

As stated in the very first sentence of Section~\ref{sec:Introduction}, our study combines two distinct realms: mechanical measurements and physical conditions.
We assume that\,---\,at any instant\,---\,the measured value of~$p$\,, in expression~\eqref{eq:formula}, is tantamount to the value of $P$ used in model~\eqref{eq:model}.
This is not the case if, for instance, pedalling is not continuous.
Nor it needs to be the case for a fixed-gear bicycle, as discussed in Section~\ref{sub:Qualifier}\,; $v_{\circlearrowright}$ might not be only an immediate consequence of $f$\,, but also of the momentum of bicycle-cyclist system.

For all measurements, we use Assioma DUO (art. 772-02) power meters, which measure power on both pedals.
The power meters use the Assioma IAV Power system~\citep{Favero2018}, which relies on a three-axis gyroscope along with proprietary algorithms to measure the angular velocity variation during each pedal stroke and, thus, results in a $\pm1\%$ accuracy rating~\citep[Technical features]{Favero2017}.
Herein, with the method used by \citeauthor{Favero2018}, $v_{\circlearrowright}$ is recorded as an instantaneous speed, not an average per revolution.
The importance of this distinction is detailed in Appendix~\ref{app:AvgPedSpeed}.

At the same time that the power-meter information is collected, the altitude and ground speed of the bicycle-cyclist system is collected using the Polar M460 GPS bike computer, which is mounted on top of the handlebars.
The GPS samples data once per second, with altitude and ascent/descent resolution rated at one and five metres, respectively~\citep[Technical Specification]{PolarM460}.
The information is collected for a flat and inclined course, for which the respective details are discussed in Sections~\ref{sec:FlatCourseVals} and~\ref{sec:InclinedCourse}.

To relate the measured power to the surrounding conditions, we mediate between the two sources of information by model~\eqref{eq:model}.
The power meter provides information regardless of the external conditions.
GPS provides partial information about the surroundings independently of the power output.
The model\,---\,based on physical principles\,---\,serves as a hypothetical relation between them.
It must remain hypothetical since there is no explicit relation between the two sources of information.
\section{Extraction of resistance coefficients}
\label{sec:NumEx}
In general, obtaining a unique result for the values of the resistance coefficients by minimizing the misfit between the right-hand side of expression~(\ref{eq:formula}), which represents measurements, and the right-hand side of model~\eqref{eq:model}, which represents a model, is impossible.
Namely, the misfit function might have several minima, with the global one not necessarily localized in the region for which the values have any physical interpretations.
For instance, $\rm C_dA$ is an area and, as such, it cannot have a negative value~\citep[e.g.,][]{Chung2012}.
Consequently\,---\,even though it is possible that a combination of model parameters that includes $\rm C_dA\leqslant0$ might lead to the best numerical result\,---\,we only consider positive values of $\rm C_dA$, $\rm C_{rr}$ and $\lambda$ to maintain a physical meaning to the results.
Moreover, the optimization relies on measurements, which are subject to experimental errors, including limitations of the GPS accuracy.

Also, different combinations of values may give the same result or the obtained values may differ between flat or hilly conditions.
To that end, we consider measurements collected on a mainly flat and nearly constant incline course in Northwestern Italy, whose relief maps and elevation plots are provided in Figure~\ref{fig:CronoMaps}.
Herein, both relief maps are sourced from~\citet{OpenStreetMap} and the altitude data, used for the elevation plots, from~\citeauthor{StravaRivaltaPontechino}.

The flat course~\citep{StravaRivaltaPontechino} consists of a 5.5-kilometre stretch on SP201, in an open valley from Rivalta Bormida to Pontechino, in Piedmont.
As indicated by the elevation plot in Figure~\ref{fig:CronoRivaltaPontechino}, the course has a little change in altitude, which is verified by comparing the small discrepancy between the average, \mbox{$\overline h = 145.4\,{\rm m}\pm5.355$\,m}\,, and median, 146.154\,m, altitudes of the course.
The inclined course~\citep{StravaRossiglioneTiglieto} consists of a 8-kilometre climb of a nearly constant inclination on SP41, between Rossiglione and Tiglieto, in Liguria.
For the elevation plot in Figure~\ref{fig:CronoRossiglioneTiglieto}, the grade of the slope of the line of best fit of altitudes is 4.3\%.
Both courses are well-maintained roads that provide open coverage for good signals and data collection.

A carbon fiber Cinelli XLE8R 5 road bike, with time-trial aero bars, is used for data collection on both courses.
As discussed in Section~\ref{sec:PhysMeas}, the power-meter measurements are acquired using Assioma Duo power meters and the GPS measurements are acquired using the Polar M460 GPS bike computer.
To increase the reliability of the parameter estimation, the cyclist remained in the saddle throughout both courses, did not use breaks, and pedalled continuously.
For the flat course, the cyclist maintained an aerodynamic position using the time-trial bars, whereas, for the inclined course, an upright position with hands on top of the handlebars.

Finally, for both courses, we set $m=111$\,kg, as the mass of the bicycle-cyclist system was the same, $w_\leftarrow=0$\,m/s, as the data collection was performed in windless conditions, and $g=9.81$\,m/s${}^2$, as both courses are in midlatitude and not at a significant altitude.

\begin{figure}
	\centering
	\begin{subfigure}[b]{0.49\textwidth}
		\centering
		\hspace*{\fill}
		\includegraphics[width=0.925\textwidth]{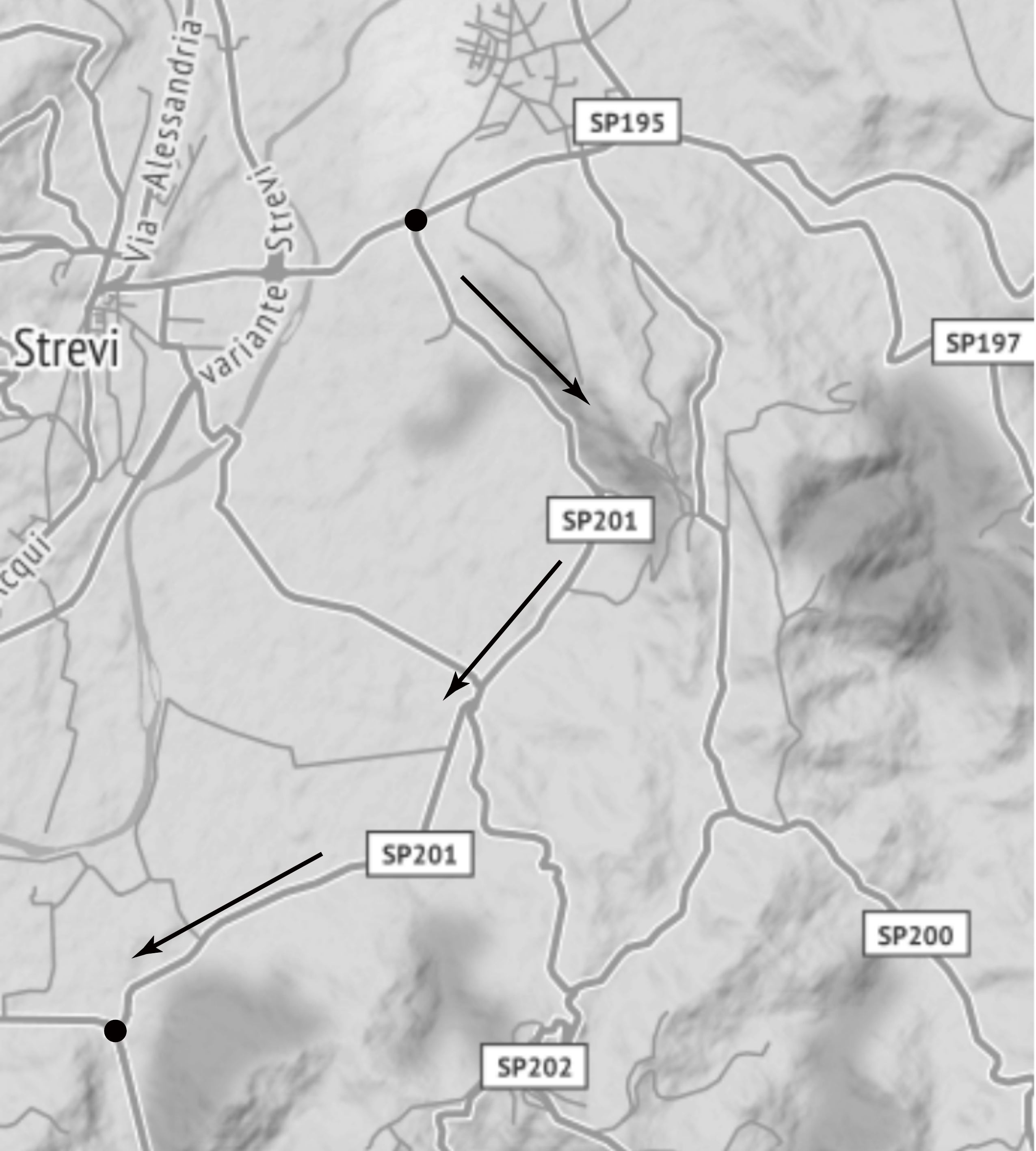}\\
		\includegraphics[width=0.95\textwidth]{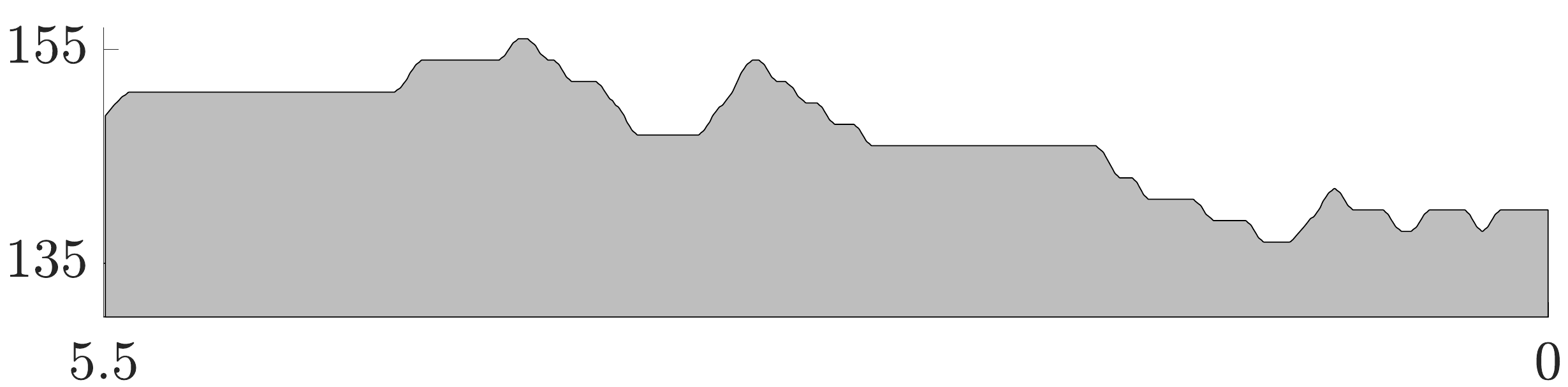}
		\caption{\small Flat course}
		\label{fig:CronoRivaltaPontechino}
	\end{subfigure}
	\hfill
	\begin{subfigure}[b]{0.49\textwidth}
		\centering
		\hspace*{\fill}
		\includegraphics[width=0.97\textwidth]{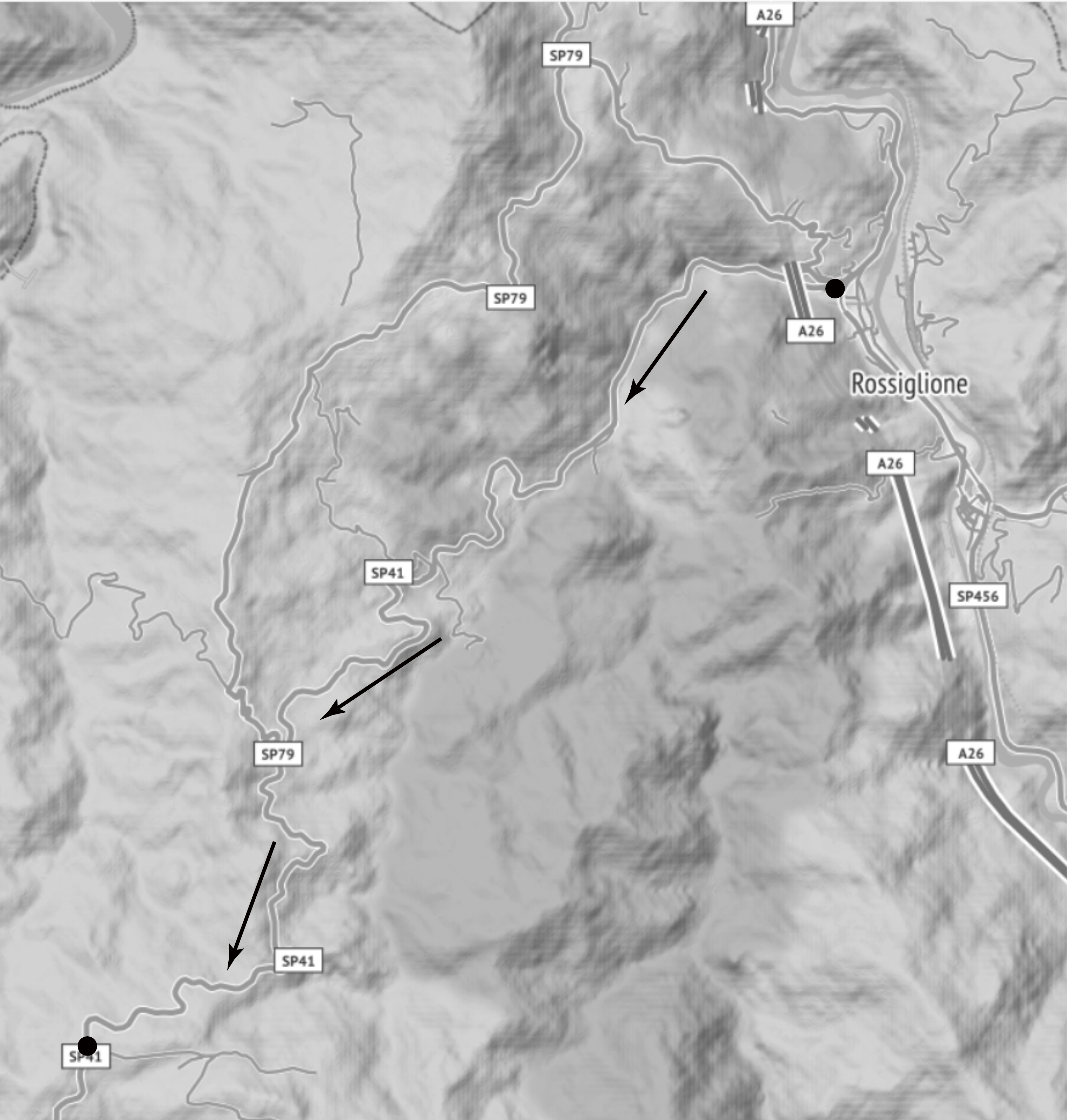}\\
		\includegraphics[width=\textwidth]{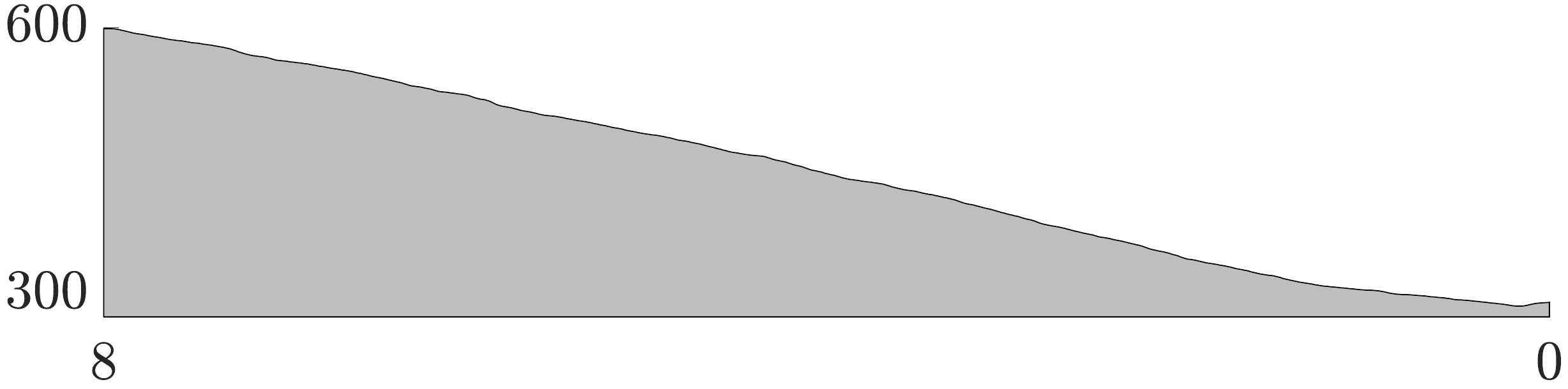}
		\caption{\small Inclined course}
		\label{fig:CronoRossiglioneTiglieto}
	\end{subfigure}
	\caption{\small
		Relief maps and elevation plots for the flat course between Rivalta Bormida and Pontechino and the inclined course between Rossiglione and Tiglieto.
		The map scales are different, and so are the scales of plot axes.}
	\label{fig:CronoMaps}
\end{figure}

\subsection{Flat-course values}
\label{sec:FlatCourseVals}
The flat-course measurements pertain to a steady ride in an aerodynamic position using time-trialing bars, beginning with a flying start.
The measurements of the force applied to and the circumferential speed of the pedals obtained during the ride are presented in Figure~\ref{fig:FigPedalVF}.
These measurements contain only the values between the first and third quartile of recorded speeds.
This restriction eliminates the data associated with the lowest and highest speeds.

\begin{figure}
	\centering
	\includegraphics[scale=0.7]{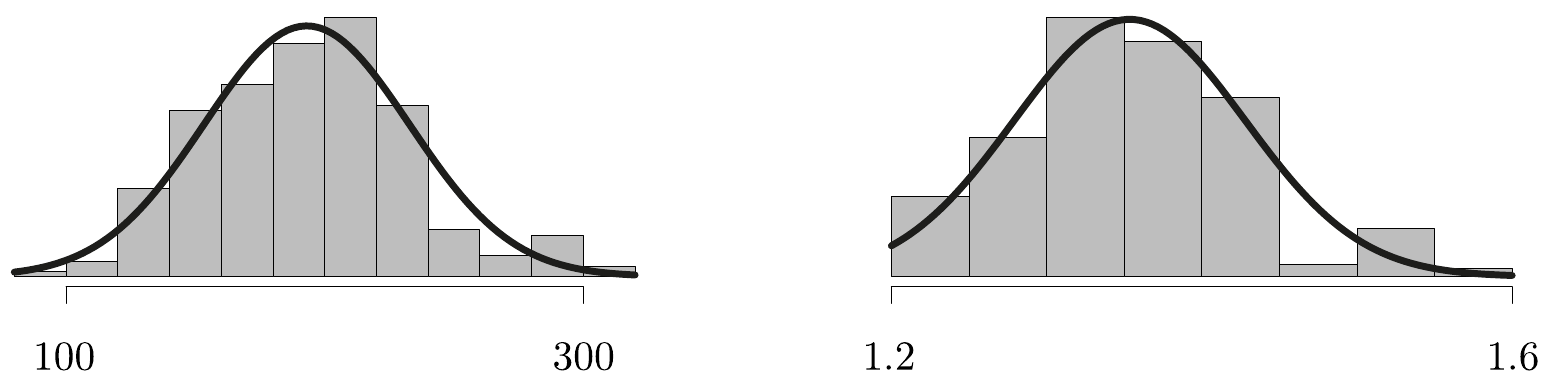}
	\caption{\small 
		Flat-course measurements. 
		Left-hand plot: force applied to pedals, $\bar f_{\circlearrowright}=193.1\,{\rm W}\pm 39.75\,{\rm W}$\,; right-hand plot: circumferential speed of pedals, $\bar v_{\circlearrowright}=1.354\,{\rm m/s}\pm 0.07448\,{\rm m/s}$
	}
	\label{fig:FigPedalVF}
\end{figure}

The power, which, according to expression~\eqref{eq:formula}, is the product of $v_{\circlearrowright}$ and $f_{\circlearrowright}$, is illustrated in the left-hand plot in Figure~\ref{fig:FigPowerVel}.
Using the ground-speed values obtained by the GPS, we group the speeds in thirty-three $0.1$-second intervals, whose centres range from $8.9$\,m/s to $12.7$\,m/s, and contain five-hundred and thirty-one speed values.
To avoid spurious results, the speed groups are restricted further to those that contain at least five values.
We present the remaining ground-speed values in the right-hand plot of Figure~\ref{fig:FigPowerVel}; the mode of the values, represented by twenty-five values, is~$9.8$\,m/s.
This approach stabilizes the search, by smoothing the measurements through averaging them over these intervals, and results in statistical information that gives an insight into the uncertainty of obtained results.

\begin{figure}
	\centering
	\includegraphics[scale=0.7]{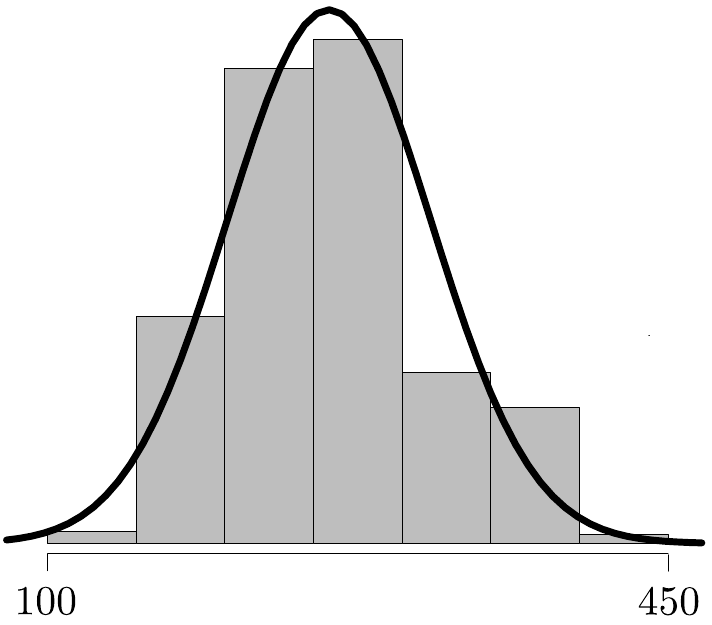}\quad
	\includegraphics[scale=0.7]{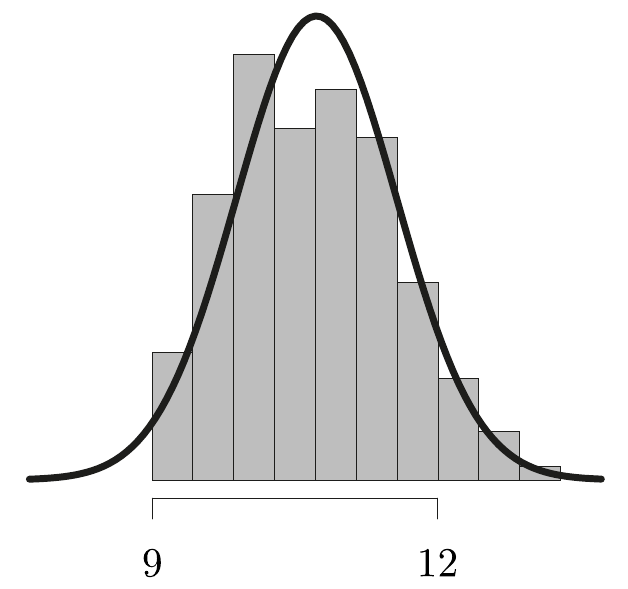}
	\caption{\small
		Flat-course measurements.
		Left-hand plot: power, from power meters, $\overline{P}=258.8\,{\rm W}\pm57.3\,{\rm W}$\,; right-hand plot: ground speed, from GPS, $\overline V_{\!\rightarrow}=10.51\,{\rm m/s}\pm0.9816\,{\rm m/s}$\,.
	}
	\label{fig:FigPowerVel}
\end{figure}

The measured power\,---\,within such a grouping\,---\,is presented in Figure~\ref{fig:FigPowerErr}; therein, standard deviations are illustrated by error bars.
To provide the values required for model~\eqref{eq:model}, we require $a$\,, $\theta$ and $\rho$\,, for each group.
They are obtained from the GPS measurements: $a$ and $\theta$ as the temporal and spatial derivatives of the measured speed and altitude, respectively, and $\rho$ by using expression~(\ref{eq:DenAlt}); these values are illustrated in Figures~\ref{fig:FigAccVel}--\ref{fig:FigDenAir}, with standard deviations illustrated by error bars.
Their average values, taken over the entire segment, are
\begin{equation}
	\label{eq:ave_a_theta_rho_flat}
	\overline{a}=0.006922\,{\rm m/s}^2\pm0.1655\,{\rm m/s}^2,\quad
	\overline{\theta} = 0.002575\,{\rm rad}\pm0.04027\,{\rm rad},\quad
	\overline\rho=1.204\,{\rm kg/m}^3\pm 0.0007652\,{\rm kg/m}^3;
\end{equation}
respectively.
They indicate a steady ride, a flat course and constant atmospheric conditions.

\begin{figure}
	\centering
	\begin{subfigure}[b]{0.475\textwidth}
		\includegraphics[width=\textwidth]{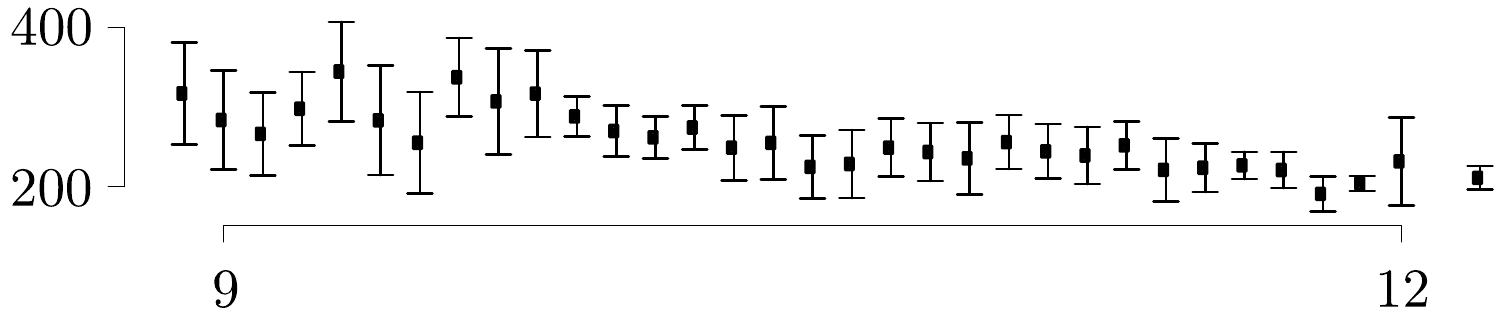}
		\caption{\small Power}
		\label{fig:FigPowerErr}
	\end{subfigure}
	\hspace*{\fill}
	\begin{subfigure}[b]{0.475\textwidth}
		\includegraphics[width=\textwidth]{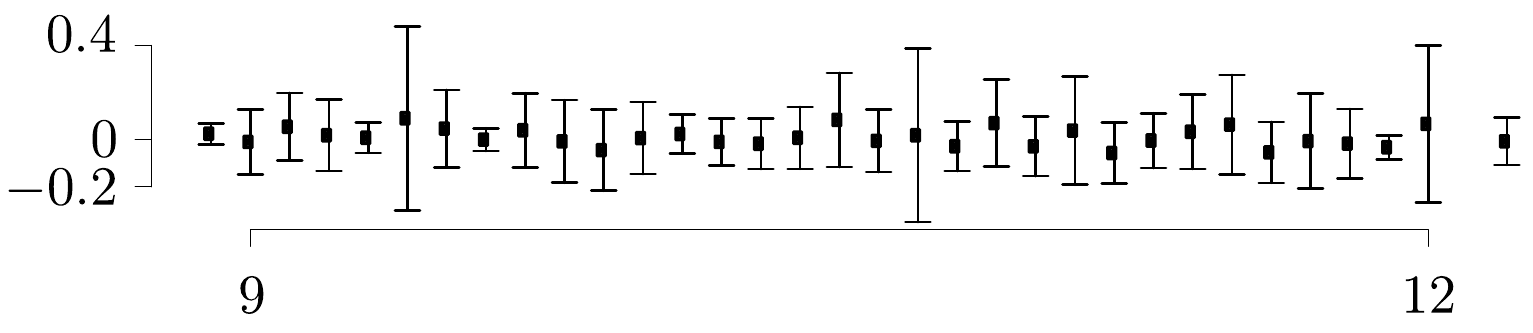}
		\caption{\small Change of speed}
		\label{fig:FigAccVel}
	\end{subfigure}
	\\
	\begin{subfigure}[b]{0.475\textwidth}
		\includegraphics[width=\textwidth]{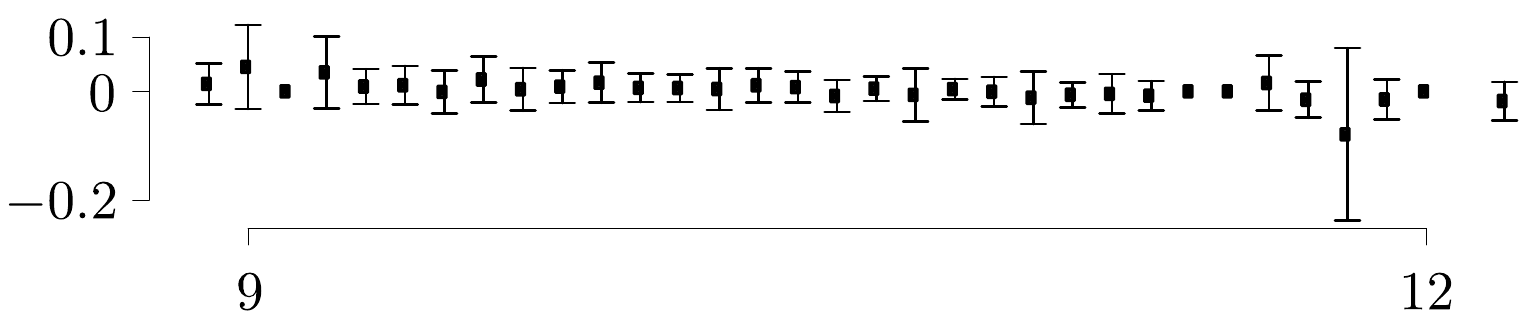}
		\caption{\small Slope}
		\label{fig:FigIncVel}
	\end{subfigure}
	\hspace*{\fill}
	\begin{subfigure}[b]{0.475\textwidth}
		\includegraphics[width=\textwidth]{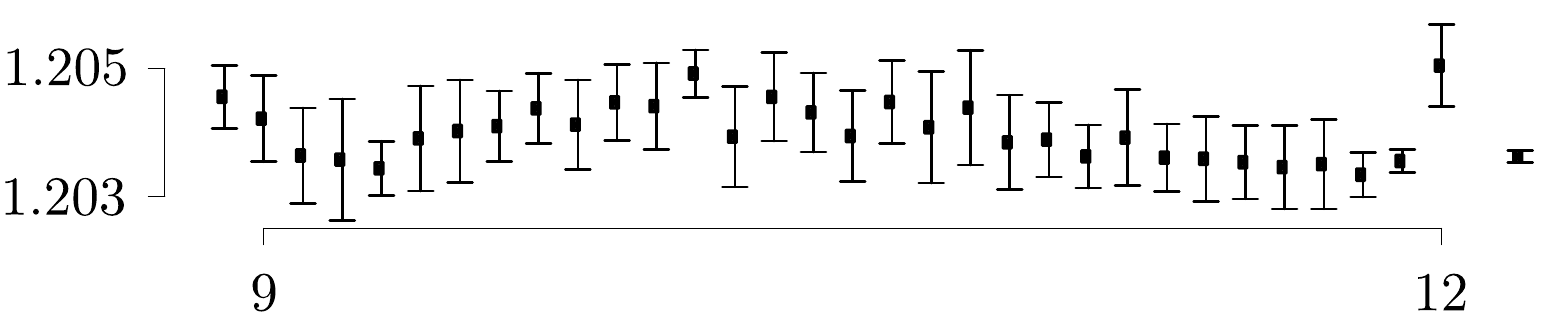}
		\caption{\small Air density}
		\label{fig:FigDenAir}
	\end{subfigure}
	\caption{\small
		Flat-course measurements of power, change of speed, slope and air density grouped into thirty-three 0.1-second intervals
	}
\end{figure}

To estimate the values of ${\rm C_{d}A}$\,, ${\rm C_{rr}}$ and $\lambda$\,, we write model~\eqref{eq:model} as
\begin{equation}
	\label{eq:misfit}
	f
	=
	P
	-
	\underbrace{
		\frac{
			mg\sin\theta
			+
			m\,a
			+
			{\rm C_{rr}}mg\cos\theta
			+
			\tfrac{1}{2}\,{\rm C_{d}A}\,\rho
			\left(V_{\!\rightarrow}+w_{\leftarrow}\right)^{2}
		}{
			1-\lambda
		}
		V_{\!\rightarrow}
	}_{F_{\!\leftarrow}V_{\!\rightarrow}}
	\,,
\end{equation}
and minimize the misfit,~$\min f$\,, using the Nelder-Mead algorithm, implemented in R.
The grouped values, with their standard deviations, are used as inputs for a local optimization.
Each group is treated separately and, hence, the statistics of its input parameters are different than for the entire set. 
In view of the expected values, a starting point for the local optimization is ${\rm C_{d}A}=0.3\,{\rm m}^2$, ${\rm C_{rr}}=0.005$ and $\lambda=0.035$.

The process is repeated ten thousand times.
The input values are perturbed in accordance with their Gaussian distributions, since\,---\,according to the central limit theorem\,---\,measurements affected by many independent processes tend to approximate such a distribution. 
We obtain optimal values with their standard deviations, 
\begin{equation*}
	{\rm C_{d}A} = 0.2607\,{\rm m}^2\pm0.002982\,{\rm m}^2
	\,,\,\,
	{\rm C_{rr}} = 0.002310\pm0.005447
	\,,\,\,
	\lambda = 0.03574\pm0.0004375
	\,,
\end{equation*}
shown in Figure~\ref{fig:FigParFlat}.
As illustrated in Figure~\ref{fig:Figfflat}, these values result in a satisfactory minimization of misfit for expression~(\ref{eq:misfit}).                     
\begin{figure}
	\centering
	\includegraphics[scale=0.7]{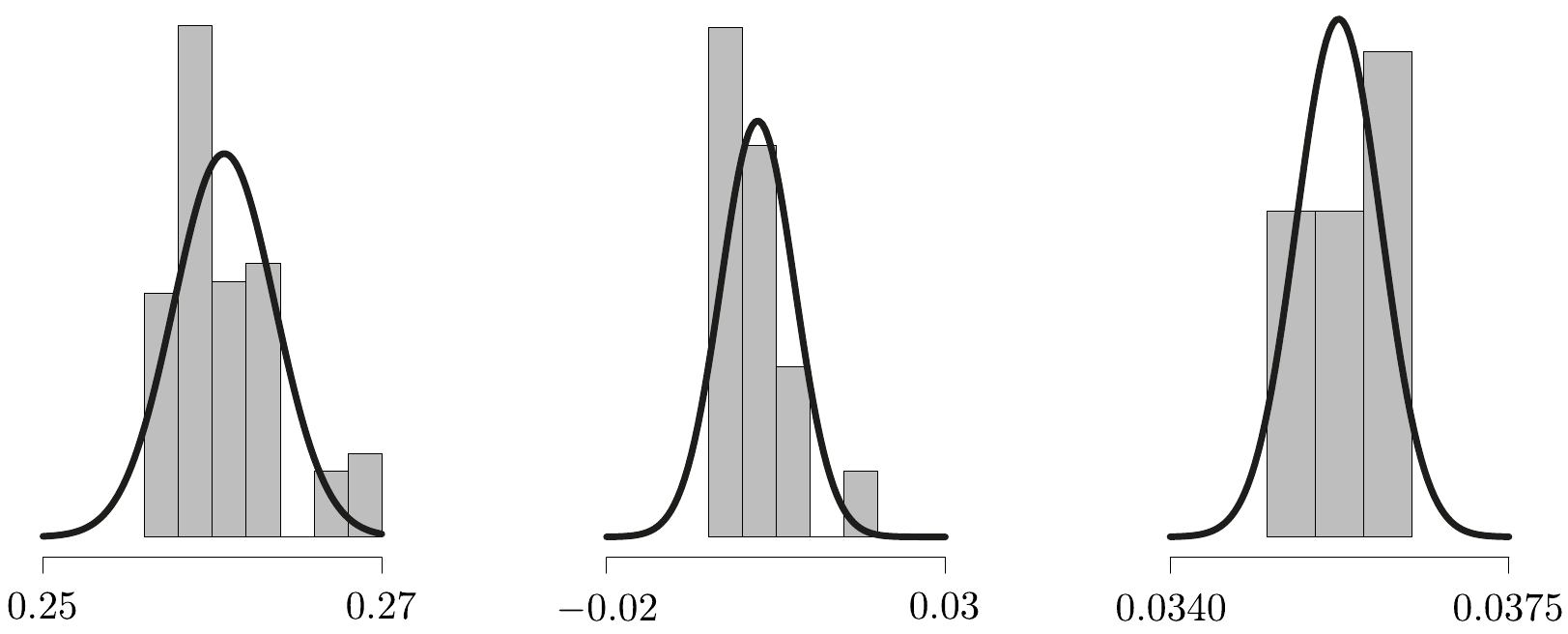}
	\caption{\small
		Flat-course optimal values; 
		left-hand plot:~${\rm C_{d}A}=0.2607\,{\rm m}^2\pm0.002982\,{\rm m}^2$\,; 
		middle plot: \mbox{${\rm C_{rr}}=0.002310\pm0.005447$}\,; 
		right-hand plot:~$\lambda=0.03574\pm0.0004375$}
	\label{fig:FigParFlat}
\end{figure}
\begin{figure}
	\centering
	\includegraphics[scale=0.7]{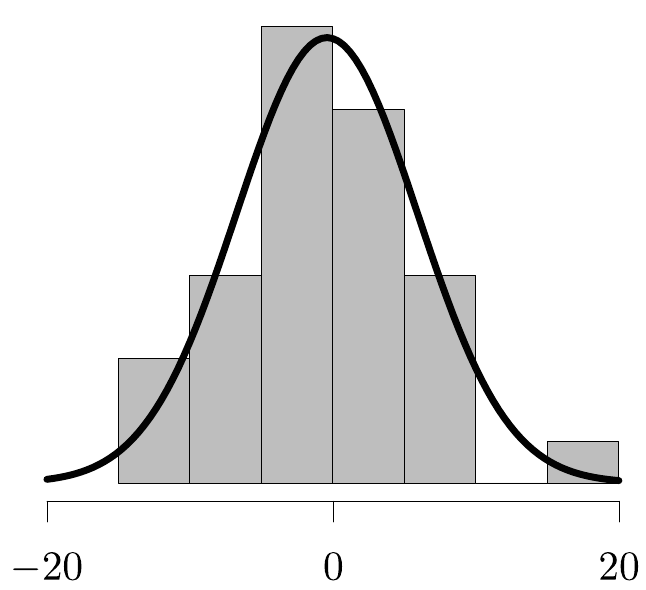}
	\caption{\small 
		Flat-course misfit of equation~(\ref{eq:misfit}): $f=0.4137\,{\rm W}\pm6.321\,{\rm W}$
	}
	\label{fig:Figfflat}
\end{figure}
Using these values, together with the average values,  over the entire segment, assuming \mbox{$\overline{a}=\overline{\theta}=0$}\,, we obtain, in accordance with model~\eqref{eq:model}, \mbox{$P = 255.3$\,W}\,, which is consistent with \mbox{$\overline{P}=258.8$\,W}\,, stated in the caption of Figure~\ref{fig:FigPowerVel}.
\subsection{Inclined-course values}
\label{sec:InclinedCourse}
The average slope of the inclined-course, recorded by the GPS, is $\overline\theta=0.04592\,{\rm rad}\pm0.1106\,{\rm rad}$, which corresponds to $2.63^\circ$ or 4.60\% and is consistent with the 4.3\% grade of the slope of the line of best fit of altitudes.
The measurements pertain to a steady ride, while remaining in the saddle throughout and hands on top of the handlebars, which is supported by an average acceleration of \mbox{$\overline a=0.001101\,{\rm m/s}^2\pm0.1015\,{\rm m/s}^2$}.
In view of the constancy of slope and steadiness of ride, we set, $\overline\theta=0.04592\,{\rm rad}$ and $\overline a = 0\,{\rm m/s^2}$\,, respectively.
Since the change of altitude is negligible\,---\,in the context of air density\,---\,we set $\overline\rho=1.168\,{\rm kg/m}^3\pm0.001861\,{\rm kg/m}^3$\,, which\,---\,under standard meteorological conditions\,---\,corresponds to the altitude of~$400$\,m\,.

For the estimation of $\rm C_dA$, $\rm C_{rr}$ and $\lambda$ on the inclined course, we apply the same restriction strategies used in for the flat course.
We restrict the force and circumferential speed measurements to contain only the values between the first and third quartile of recorded speed, and present their product, which is the power, in the left-hand plot of Figure~\ref{fig:FigSpeedPowerSteep}.
Also, we apply the same ground-speed grouping strategy, which results in eleven~$0.1$-second speed intervals, whose centres range from $3.7$\,m/s to $4.7$\,m/s, and contain three-hundred and ninety-two values.
We present these speeds in the right-hand plot of Figure~\ref{fig:FigSpeedPowerSteep}, wherein the mode is $4.2$\,m/s, represented by seventy-eight data points; their respective means are $\overline V_{\!\rightarrow}=4.138\,{\rm m/s}\pm0.2063\,{\rm m/s}$ and $\overline P=286.6\,{\rm W}\pm33.11\,{\rm W}$.
Proceeding with the same optimization as the flat course, using the same starting values as input, with random perturbations within their standard deviations, the process is repeated ten thousand times to obtain the distribution of \mbox{$\rm C_dA$, $\rm C_{rr}$ and $\lambda$}, which we illustrate in Figure~\ref{fig:FigParSteep}; their values are
\begin{equation}
	{\rm C_dA} = 0.2702\,{\rm m}^2\pm0.002773\,{\rm m}^2,\quad
	{\rm C_{rr}}=0.01298\pm0.011,\quad
	\lambda=0.02979\pm0.004396.
\end{equation}

\begin{figure}
	\centering
	\includegraphics[scale=0.7]{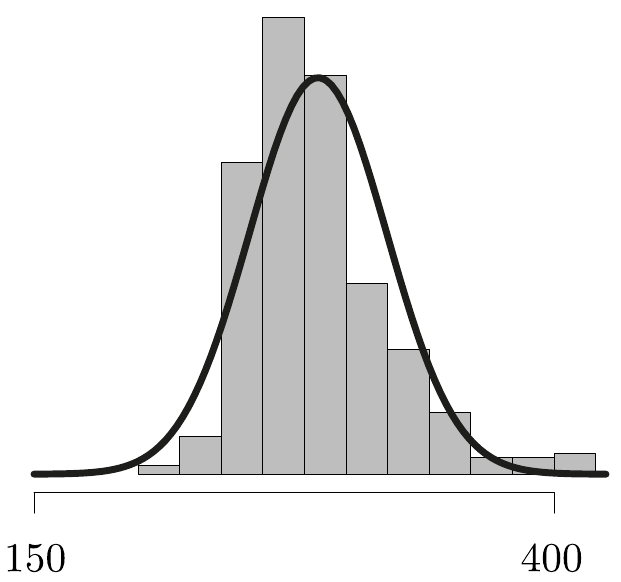}\qquad
	\includegraphics[scale=0.7]{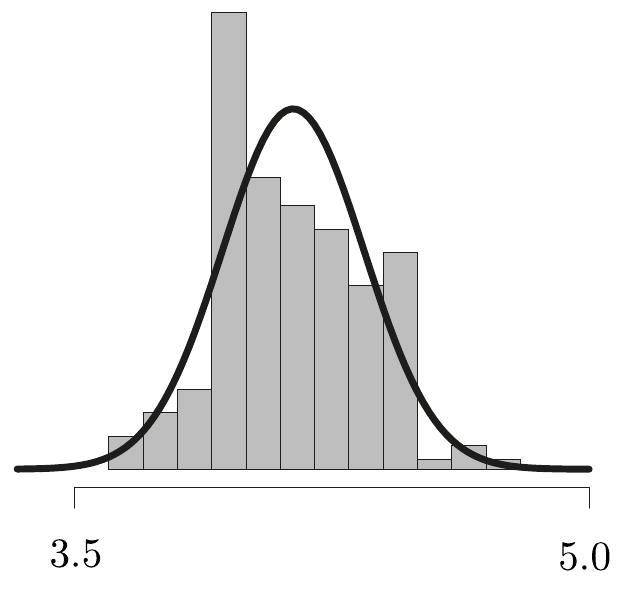}
	\caption{\small
		Inclined-course measurements.
		Left-hand plot: power, from power meters, $\overline{P}=286.6\,{\rm W}\pm33.07\,{\rm W}$\,;
		right-hand plot: ground speed, from GPS measurements, $\overline V_{\!\rightarrow}=4.138\,{\rm m/s}\pm0.2063\,{\rm m/s}$\,.
	}
\label{fig:FigSpeedPowerSteep}
\end{figure}
\begin{figure}
	\centering
	\includegraphics[scale=0.7]{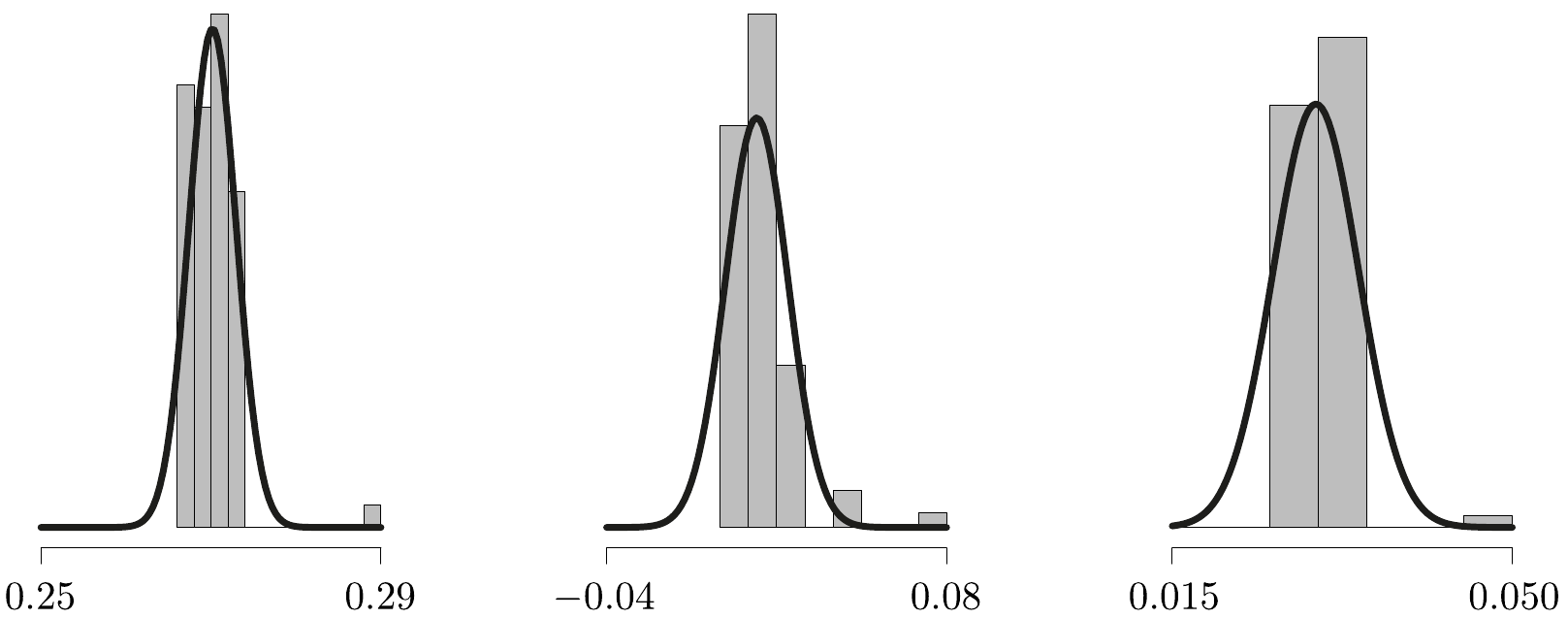}
	\caption{\small
		Inclined-course optimal values.
		left-hand plot:~${\rm C_{d}A}=0.2702\,{\rm m}^2\pm0.002773\,{\rm m}^2$\,;
		middle plot:~\mbox{${\rm C_{rr}}=0.01298\pm0.01100$}\,;
		right-hand plot:~$\lambda=0.02979\pm0.004396$}
	\label{fig:FigParSteep}
\end{figure}

\section{Rates of change of model parameters}
\label{sec:RatesOfChange}
\subsection{Implicit function theorem}
\label{sec:ImplFuncThm}
We seek relations between the ratios of quantities on the right-hand side of model~\eqref{eq:model}.
To do so\,---\,since $f$\,, stated in expression~\eqref{eq:misfit}, possesses continuous partial derivatives in all its variables at all points, except at ${\lambda=1}$\,, which is excluded by mechanical considerations, and since ${f=0}$\,, as a consequence of equation~(\ref{eq:formula})\,---\,we invoke the implicit function theorem to write
\begin{equation}
	\label{eq:Thm}
	\dfrac{\partial y}{\partial x}
	=
	-\dfrac{
		\dfrac{\partial f}{\partial x}
	}{
		\dfrac{\partial f}{\partial y}
	}
	=:
	-\dfrac{\partial_{x}f}{\partial_{y}f}
	\,,
\end{equation}
where $x$ and $y$ are any two quantities among the arguments of
\begin{equation}
	\label{eq:f}
	f(P,m,g,\theta,a,{\rm C_{rr}},{\rm C_{d}A},\rho,V_{\!\rightarrow},w_{\leftarrow},\lambda)
	\,.
\end{equation}

To use formula~(\ref{eq:Thm}), in the context of expression~\eqref{eq:misfit}, we obtain all partial derivatives of $f$\,, with respect to its arguments:
\begin{equation*}
	\partial_{P}f = 1
	,\quad
	\partial_{m}f 
	= 
	-\dfrac{
		a 
		+ 
		g\left({\rm C_{rr}}\cos\theta+\sin\theta\right)
	}{
		1-\lambda
	}
	V_{\!\rightarrow}
	\,,\quad
	\partial_{\theta}f 
	= 
	-\dfrac{
		m\,g\left(\cos\theta-{\rm C_{rr}}\sin\theta\right)
	}{
		1-\lambda
	}
	V_{\!\rightarrow}
	\,,\quad
	\partial_{a}f = -\dfrac{m\,V_{\!\rightarrow}}{1-\lambda}
	\,,
\end{equation*}
\begin{equation*}
	\partial_{\rm C_{rr}}f 
	= 
	-\dfrac{m\,g\cos\theta}{1-\lambda}V_{\!\rightarrow}
	\,,\quad
	\partial_{\rm C_{d}A}f
	= 
	-\dfrac{
		\rho
		\left(V_{\!\rightarrow}+w_{\leftarrow}\right)^{2}
	}{
		2\,(1-\lambda)
	}
	V_{\!\rightarrow}
	\,,\quad
	\partial_{\rho}f 
	= 
	-\dfrac{
		{\rm C_{d}A}
		\left(V_{\!\rightarrow}+w_{\leftarrow}\right)^{2}
	}{
		2\,(1-\lambda)
	}
	V_{\!\rightarrow}
	\,,
\end{equation*}
\begin{equation*}
	\partial_{V_{\!\rightarrow}}f
	=
	-\dfrac{
		2\,m\left(a+g\left({\rm C_{rr}}\cos\theta+\sin\theta\right)\right)
		+
		\rho\,{\rm C_{d}A}
		\left(V_{\!\rightarrow}+w_{\leftarrow}\right)
		\left(3\,V_{\!\rightarrow}+w_{\leftarrow}\right)
	}{
		2\,(1-\lambda)
	}
	\,,
\end{equation*}
\begin{equation*}
	\partial_{w_{\leftarrow}}f
	=
	-\dfrac{
		\rho\,{\rm C_{d}A}
		\left(V_{\!\rightarrow}+w_{\leftarrow}\right)
	}{
		1-\lambda
	}
	V_{\!\rightarrow}
	\,,\quad
	\partial_{\lambda}f 
	=
	-\dfrac{
		m\left(a+g\left({\rm C_{rr}}\cos\theta+\sin\theta\right)\right)
		+
		\tfrac{1}{2}\rho\,{\rm C_{d}A}
		\left(V_{\!\rightarrow}+w_{\leftarrow}\right)^{2}
	}{
		(1-\lambda)^2
	}
	V_{\!\rightarrow}
	\,.
\end{equation*}
In accordance with the definition of a partial derivative, all variables in expression~(\ref{eq:f}) are constant, except the one with respect to which the differentiation is performed.
This property is apparent in Appendix~\ref{app:LagrangeMultipliers}, where we examine a relation between differences and derivatives.
\renewcommand{\arraystretch}{1.1}
\begin{table}[h]
	\centering
	\begin{tabular}{cccc}
		\toprule
		Partial derivative & Units & Flat course & Inclined course \\
		\bottomrule\\[-1em]
		$\partial_{P}f$ & --- & 1 & 1
		\\
		$\partial_{m}f$ & ${\rm W}\,{\rm kg}^{-1}$ &
		$-0.2470\,\pm\,0.5829$ & $-2.463\,\pm\,4.644$ 
		\\
		$\partial_{\theta}f$ & W &
		$-11870\,\pm\,1109$ & $-4637\,\pm\,234.1$
		\\
		$\partial_a f$ & ${\rm W}\,{\rm m}^{-1}\,{\rm s}^2$ &
		$-1210\,\pm\,113.0$ & $-473.4\,\pm\,23.70$
		\\
		$\partial_{\rm C_{rr}}f$ & W &
		$-11870\,\pm\,1109$ & $-4639\,\pm\,233.4$
		\\
		$\partial_{\rm C_{d}A}f$ & ${\rm W}\,{\rm m}^{-2}$ &
		$-724.8\,\pm\,203.2$ & $-42.66\,\pm\,6.384$
		\\
		$\partial_{\rho}f$ & ${\rm W}\,{\rm kg}^{-1}\,{\rm m}^3$ &
		$-157.0\,\pm\,44.01$ & $-9.866\,\pm\,1.480$
		\\
		$\partial_hf$ & W &
		$0.02240\,\pm\,0.006282$ & $0.001367\,\pm\,0.0002050$
		\\
		$\partial_{V_{\!\rightarrow}}f$ & ${\rm W}\,{\rm m}^{-1}\,{\rm s}$ &
		$-56.55\,\pm\,11.82$  & $-74.43\,\pm\,124.5$
		\\
		$\partial_{w_{\leftarrow}} f$ & ${\rm W}\,{\rm m}^{-1}\,{\rm s}$ &
		$-35.96\,\pm\,6.729$ & $-5.571\,\pm\,0.5591$
		\\
		$\partial_{\lambda}f$ & W &
		$-224.4\,\pm\,88.39$ & $-293.7\,\pm\,531.4$
		\\
		\bottomrule
	\end{tabular}
	\caption{\small Values of partial derivatives for flat and inclined courses, rounded to four significant figures}
	\label{table:Partials}
\end{table}
\renewcommand{\arraystretch}{1}
\subsection{Interpretation}
\label{sec:Interpretation}
\begin{figure}
	\centering
	\includegraphics[scale=0.7]{Figfflat.pdf}
	\includegraphics[scale=0.7]{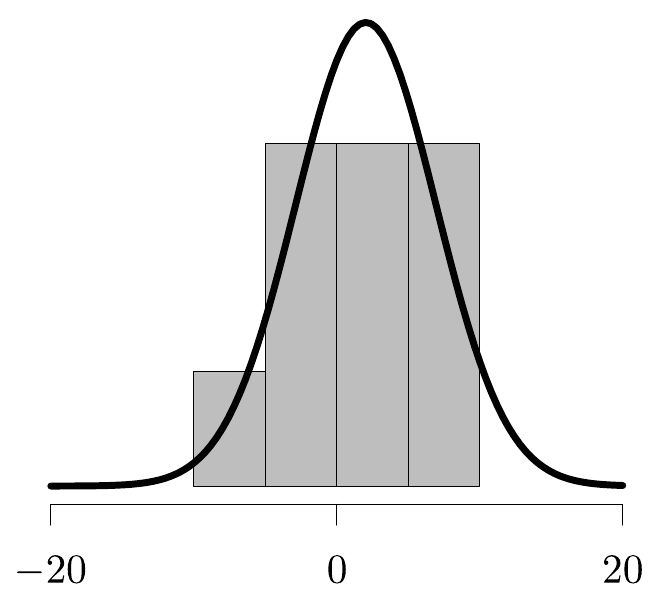}
	\caption{\small Misfit of equation~\eqref{eq:misfit}: left-hand plot: flat course, $f=0.4137\,{\rm W}\pm6.321\,{\rm W}$\,; right-hand plot: inclined course, $f=2.03\,{\rm W}\pm4.911\,{\rm W}$}
	\label{fig:Figf}
\end{figure}
\subsubsection{Preliminary comments}
To use the partial derivatives tabulated in Table~\ref{table:Partials}, we consider the measurements collected as well as optimized values from the two rides discussed in Sections~\ref{sec:FlatCourseVals} and~\ref{sec:InclinedCourse}.

As required by the implicit function theorem and as shown in Figure~\ref{fig:Figf}, $f=0$\,, in the neighbourhood of the maxima of the distributions, for both the flat and inclined courses.
Also, as required by the theorem, in formula~(\ref{eq:Thm}), and as shown in Table~\ref{table:Partials}, $\partial_{y}f\neq0$\,, in the neighbourhoods of interest, for either course.

Notably, the similarity of a horizontal spread for both plots of Figure~\ref{fig:Figf} indicates that the goodness of fit of a model is similar for both courses.
The spread is slightly narrower for the inclined course; this might be a result of a lower average speed,~$\overline V_{\!\rightarrow}$\,, which allows for more data points for a given distance and, hence, a higher accuracy of information.
\subsubsection{Model inferences}
\label{sec:ModelInferences}
The misfit minimization of equation~\eqref{eq:misfit}, $\min f$\,, treats $\rm C_{d}A$\,, $\rm C_{rr}$ and $\lambda$ as adjustable parameters.
The values in Table~\ref{table:ModelRates} are the changes of $\rm C_{d}A$ due to a change in $\rm C_{rr}$ or $\lambda$\,; in either case, the other quantities are kept constant.
Let us examine the first row.

For the flat course\,---\,in the neighbourhood of $\overline V_{\!\rightarrow}=10.51\,{\rm m/s}$ and $\overline P=258.8\,{\rm W}$\,, wherein ${\rm C_{d}A}=0.2607\,{\rm m}^2$ and ${\rm C_{rr}}=0.002310$\,---\,$\partial_{\rm C_{rr}}{\rm C_{d}A}=-16.37\,{\rm m}^2$ and, in accordance with expression~(\ref{eq:Thm}), its reciprocal is \mbox{$\partial_{\rm C_{d}A}{\rm C_{rr}}=-0.06107\,{\rm m}^{-2}$}.
We write the corresponding differentials as
\begin{equation*}
	{\rm d(C_{d}A)}
	=
	\dfrac{\partial{\rm C_{d}A}}{\partial{\rm C_{rr}}}\,
	{\rm d(C_{rr})}
	=
	{-16.37}\,{\rm d(C_{rr})}
	\quad\text{and}\quad
	{\rm d(C_{rr})}
	=
	\dfrac{\partial{\rm C_{rr}}}{\partial{\rm C_{d}A}}\,
	{\rm d(C_{d}A)}
	=
	{-0.06107}\,{\rm d(C_{d}A)}
	\,;
\end{equation*}
in other words, an increase of $\rm C_{rr}$ by a unit corresponds to a decrease of  $\rm C_{d}A$ by $16.37$ units, and an increase of $\rm C_{d}A$ by a unit corresponds to a decrease of  $\rm C_{rr}$ by ${0.06107}$ of a unit.
For the inclined course\,---\,in the neighbourhood of $\overline V_{\!\rightarrow}=4.138\,{\rm m/s}$ and $\overline P=286.6\,{\rm W}$\,, wherein ${\rm C_{d}A}=0.2702\,{\rm m}^2$ and ${\rm C_{rr}}=0.01298$\,---\,the differentials are
\begin{equation*}
	{\rm d(C_{d}A)}
	=
	\dfrac{\partial{\rm C_{d}A}}{\partial{\rm C_{rr}}}\,
	{\rm d(C_{rr})}
	=
	{-0.009195}\,{\rm d(C_{rr})}
	\quad\text{and}\quad
	{\rm d(C_{rr})}
	=
	\dfrac{\partial{\rm C_{rr}}}{\partial{\rm C_{d}A}}\,
	{\rm d(C_{d}A)}
	=
	{-108.8}\,{\rm d(C_{d}A)}
	\,;
\end{equation*}

Remaining within a linear approximation, an increase of $\rm C_{d}A$ by $1\%$ corresponds to a decrease of $\rm C_{rr}$ by $6.89\%$\,, for the flat course, and a decrease of only $0.19\%$\,, for the inclined course.
This result quantifies that the dependence between $\rm C_{d}A$ and $\rm C_{rr}$\,, within adjustments of the model, is more pronounced for the flat course than for the inclined course, as expected in view of model~\eqref{eq:model}, whose value\,---\,for the inclined course\,---\,is dominated by the first summand in the numerator, which includes neither $\rm C_{d}A$ nor $\rm C_{rr}$\,.
This result provides a quantitative justification for the observation that the dependance of the accuracy of the estimate of power on the accuracies of $\rm C_{d}A$ and $\rm C_{rr}$ varies depending on the context; it is more pronounced on flat and fast courses.

Similar evaluations can be performed using the values of derivatives contained in the second row of Table~\ref{table:ModelRates}.
Therein, an increase in $\lambda$ results in a decrease of $\rm C_{d}A$\,, with different rates, for the flat and inclined courses.
\renewcommand{\arraystretch}{1.1}
\begin{table}
	\centering
	\begin{tabular}{cccc}
	\toprule
	Partial derivative & Units & Flat course & Inclined course \\
	\bottomrule\\[-1em]
	$\partial_{\rm C_{rr}}{\rm C_{d}A}$ & ${\rm m}^2$ & 
	$-16.37\,\pm\,3.059$ & $-108.8\,\pm\,10.86$ 
	\\
	$\partial_\lambda{\rm C_{d}A}$ & ${\rm m}^2$ &
	$-0.3096\,\pm\,0.09284$ & $-6.884\,\pm\,12.47$
	\\
	\bottomrule
	\end{tabular}
	\caption{\small Model rates of change following formula~(\ref{eq:Thm}) and values, rounded to four significant figures, in Table~\ref{table:Partials}.
	}
	\label{table:ModelRates}
\end{table}
\renewcommand{\arraystretch}{1}
\subsubsection{Physical inferences}
\label{sec:PhysicalInferences}
Physical inferences\,---\,based on minimization of expression~\eqref{eq:misfit}\,---\,are accurate in a neighbourhood of $\overline V_{\!\rightarrow}$ and $\overline P$\,, wherein the set of values for ${\rm C_{d}A}$\,, ${\rm C_{rr}}$ and $\lambda$ is estimated, since, as discussed in Section~\ref{sec:ModelInferences}, these values\,---\,in spite of their distinct physical interpretations\,---\,are related among each other by the process of optimization of the model.

In view of model~\eqref{eq:model}, power as a function of ground speed is cubic.
As stated in Proposition~\ref{prop:one-to-one}, below, and shown in Figure~\ref{fig:FigPowerSpeed}, for a steady ride along a flat or an uphill, the relationship between power and speed is one-to-one and results in a concave-up curve for~$V_{\!\rightarrow}>0$\,.
\begin{figure}
	\centering
	\includegraphics[scale=0.7]{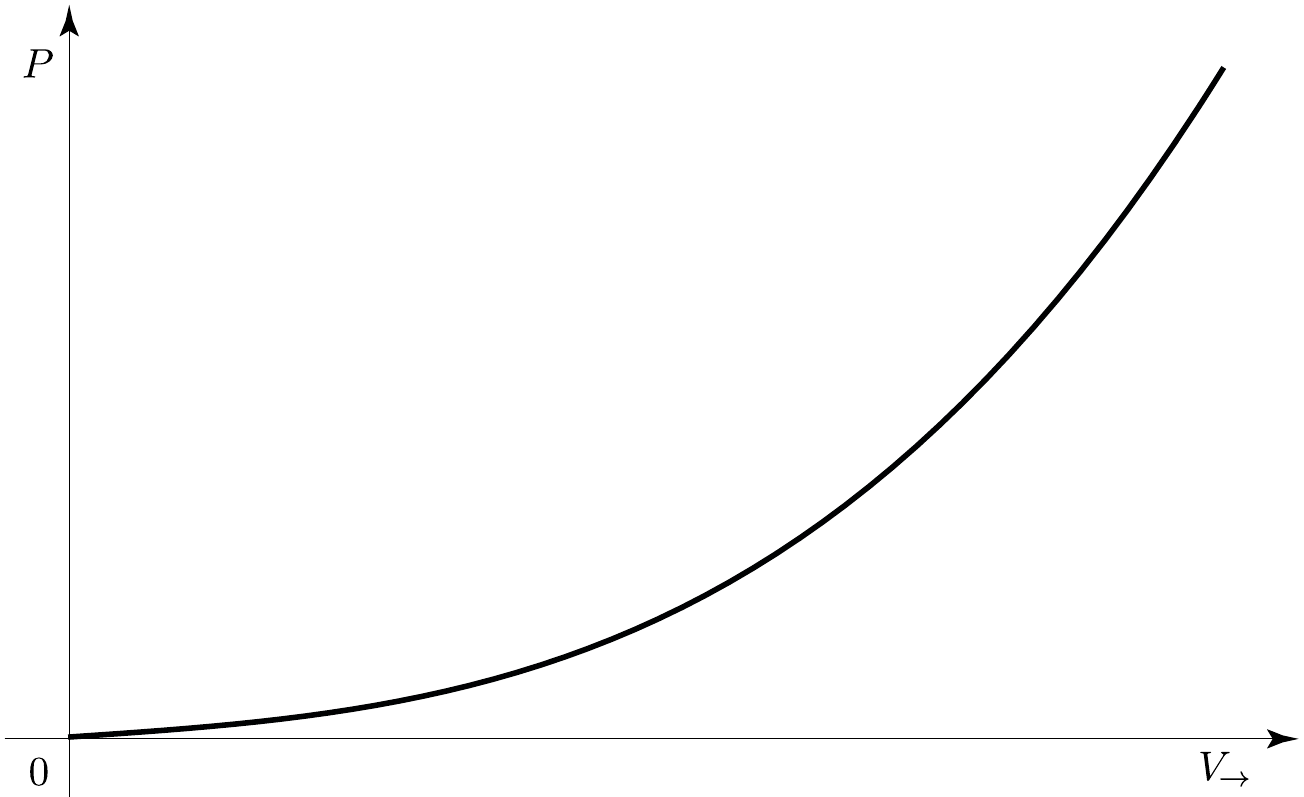}
	\caption{\small 
		Power as a function of speed}
	\label{fig:FigPowerSpeed}
\end{figure}

The values in Table~\ref{table:PhysicalRates} are the changes of ground speed due to a change in power, mass, slope and wind; in each case, the other quantities are kept constant.
These values allow us to answer such questions as what is the required increase of power to increase speed by 1 metre per second?
To answer this question, let us examine the first row.

Corresponding to their respective neighbourhoods and reciprocals from Table~\ref{table:PhysicalRates}, we have \mbox{$\partial_{V_{\!\rightarrow}}P=56.55\,{\rm W}\,{\rm m}^{-1}\,{\rm s}$} for the flat course and \mbox{$\partial_{V_{\!\rightarrow}}P=74.43\,{\rm W}\,{\rm m}^{-1}\,{\rm s}$} for the inclined course.
Thus, an increase of $V_{\!\rightarrow}$ by 1~metre per second requires an increase of $P$ of about 57~watts for the former and 74~watts for the latter.%
\footnote{$\partial P/\partial V$ can be also found by differentiating model~\eqref{eq:model} with respect to $V_{\!\rightarrow}$\,.
However, the implicit function theorem allows us to obtain relations between quantities, without explicitly expressing one in terms of the other.
Also, in accordance with the inverse function theorem, $\partial V/\partial P=1/(\partial P/\partial V)$\,, which is justified by the fact that, for model~\eqref{eq:model}, $\partial P/\partial V_{\!\rightarrow}\neq 0$\,, in the neighbourhood of interest, as required by the theorem.
However, expression~(\ref{eq:Thm}), which states the implicit function theorem, provides a convenience of examining the relations between the rates of change of any two quantities without invoking the inverse function theorem and requiring an explicit expression for either of them.}
However, this speed increase corresponds to just 9.5\% on the flat course, but 24\% on the inclined course.
Thus, remaining within a linear approximation, an 1\%-increase of speed requires about a 2.3\%-increase in power on the former course, but only 1.1\% on the latter.
This result provides a quantitative justification for a time-trial adage of pushing on the uphills and recovering on the flats, to diminish the overall time.

Since, as illustrated in Figure~\ref{fig:FigPowerSpeed}, the slope of the tangent line changes along the curve, the value of expression~(\ref{eq:Thm}) corresponds to a given neighbourhood of $(\overline V_{\!\rightarrow},\overline P)$ pairs.
Our interpretation is tantamount to comparing the slopes of two such curves\,---\,one corresponding  to the model of the flat course and the other of the inclined course\,---\,at two distinct locations, $(\overline V_{\!\rightarrow},\overline P)=(10.51,258.8)$ and $(\overline V_{\!\rightarrow},\overline P)=(4.138,286.6)$\,.
Even though the slope of the tangent line changes, it is positive for all values.
This means that the function is monotonically increasing, even though it is a third degree polynomial.
In other words, the relation of power and speed is a bijection, as illustrated in Figure~\ref{fig:FigPowerSpeed} and as discussed in Section~\ref{sec:one-to-one}.

Now, let us consider another question: remaining within a linear approximation, what percentage decrease in mass is required to result in a 1\%-increase in speed for both the flat and inclined courses?
To do so, we examine the second row of Table~\ref{table:PhysicalRates}, where, for the flat course, $\partial_{m}V_{\!\rightarrow}=-0.004368\,{\rm m}\,{\rm s}^{-1}\,{\rm kg}^{-1}$ and, for the inclined course, $\partial_{m}V_{\!\rightarrow}=-0.03309\,{\rm m}\,{\rm s}^{-1}\,{\rm kg}^{-1}$.
Using the corresponding differentials, the 1\%-increase in speed requires a decrease of mass of about 22\% on the flat, but only 1\% on the incline.
This is supportive evidence of an empirical insight into the importance of lightness for climbing; in contrast to flat courses, in the hills, even a small loss of weight results in a noticeable advantage.
Also, this result can be used to quantify the importance of the power-to-weight ratio, which plays an important role in climbing, but a lesser one on a flat.

Similar evaluations can be performed using the values of derivatives contained in the third and fourth rows of Table~\ref{table:PhysicalRates}.
In both cases, the sign is negative; hence, as expected, the increase of steepness or headwind results in a decrease of speed.
These rates of decrease, which are different for the flat and inclined courses, can be quantified in a manner analogous to the one presented in this section.
\renewcommand{\arraystretch}{1.1}
\begin{table}
	\centering
	\begin{tabular}{cccc}
	\toprule
	Partial derivative & Units & Flat course & Inclined course \\
	\bottomrule\\[-1em]
	$\partial_{P}V_{\!\rightarrow}$ & ${\rm m}\,{\rm s}^{-1}\,{\rm W}^{-1}$ &
	$0.01768\,\pm\,0.003697$ & $0.01344\,\pm\,0.02248$
	\\
	$\partial_{m}V_{\!\rightarrow}$ & ${\rm m}\,{\rm s}^{-1}\,{\rm kg}^{-1}$ &
	$-0.004368\,\pm\,0.009832$ & $-0.03309\,\pm\,0.007120$
	\\
	$\partial_{\theta}V_{\!\rightarrow}$ & ${\rm m}\,{\rm s}^{-1}$ &
	$-209.9\,\pm\,29.04$ & $-62.30\,\pm\,104.7$
	\\
	$\partial_{w_{\leftarrow}}V_{\!\rightarrow}$ & --- &
	$-0.6359\,\pm\,0.06939$ & $-0.07485\,\pm\,0.1254$
	\\
	\bottomrule
	\end{tabular}
	\caption{\small Physical rates of change following formula~(\ref{eq:Thm}) and values, rounded to four significant figures, in Table~\ref{table:Partials}.
	}
	\label{table:PhysicalRates}
\end{table}
\renewcommand{\arraystretch}{1}
\section{Model considerations}
\subsection{Relation between power and speed}
\label{sec:one-to-one}
\begin{proposition}
\label{prop:one-to-one}
According to model~\eqref{eq:model}, with $a=0$\,, $V_{\!\rightarrow}>w_{\leftarrow}$ and $0\leqslant\theta\leqslant\pi/2$\,, the relation between the measured power,~$P$\,, and the bicycle speed,~$V_{\!\rightarrow}>0$\,, is one-to-one.
Also, $P$ as a function of $V_{\!\rightarrow}$ is concave-up.
\end{proposition}
\begin{proof}
It suffices to show that $\partial P/\partial V_{\!\rightarrow}>0$ and $\partial^2P/\partial V_{\!\rightarrow}^2>0$\,, for $a=0$\,, $V_{\!\rightarrow}>w_{\leftarrow}$\,, $0\leqslant\theta\leqslant\pi/2$ and $V_{\!\rightarrow}\in(0,\infty)$\,, where
\begin{equation*}
	P=\frac{(\sin\theta+{\rm C_{rr}}\cos\theta)\,m\,g\,V_{\!\rightarrow}+\tfrac{1}{2}\,{\rm C_{d}A}\,\rho\,\left(V_{\!\rightarrow}^3+2\,w_{\leftarrow}\,V_{\!\rightarrow}^2+w_{\leftarrow}^2\,V_{\!\rightarrow}\right)}{1-\lambda}\,.
\end{equation*}
The first and second derivatives are
\begin{equation*}
	\dfrac{\partial P}{\partial V_{\!\rightarrow}}
	=\frac{(\sin\theta+{\rm C_{rr}}\cos\theta)\,m\,g+\tfrac{1}{2}\,{\rm C_{d}A}\,\rho\,\left(3V_{\!\rightarrow}^2+4\,w_{\leftarrow}\,V_{\!\rightarrow}+w_{\leftarrow}^2\right)}{1-\lambda}
	\quad{\rm and}\quad
	\dfrac{\partial^2P}{\partial V_{\!\rightarrow}^2}
	=\frac{{\rm C_{d}A}\,\rho\,\left(3\,V_{\!\rightarrow}+2\,w_{\leftarrow}\right)}{1-\lambda}\,.
\end{equation*}
Since, in physically meaningful situations, $\lambda<1$\,, a sufficient condition for a one-to-one relation is for both summands in the numerator of the first derivative,~$\partial P/\partial V_{\!\rightarrow}$\,, to be positive.
Since, $m>0$\,, $g>0$\,, $\rm C_{rr}>0$ and $\sin\theta\geqslant0$\,, $\cos\theta\geqslant0$\,, for $0\leqslant\theta\leqslant\pi/2$\,, the first summand is positive.
Since ${\rm C_{d}A}>0$\,, $\rho>0$\,, the second summand is positive if and only if
\begin{equation*}
3V_{\!\rightarrow}^2+4\,w_{\leftarrow}\,V_{\!\rightarrow}+w_{\leftarrow}^2> 0\,.
\end{equation*}
Considering a headwind\,---\,for which $w_{\leftarrow}$ is positive\,---\,we obtain $V_{\!\rightarrow}< -w_{\leftarrow}$ and $V_{\!\rightarrow}> -w_{\leftarrow}/3$\,.
Since both correspond to $V_{\!\rightarrow}<0$\,, there are outside the domain of interest.
Hence, for any headwind, the second summand is positive.

Considering a tailwind\,---\,for which $w_{\leftarrow}$ is negative\,---\,we obtain $V_{\!\rightarrow}< w_{\leftarrow}/3$ and $V_{\!\rightarrow}>w_{\leftarrow}$\,.
To include the entire domain of $V_{\!\rightarrow}>0$\,, we consider only the latter interval: the ground speed must be greater than the tailwind.

For convexity, we need $3\,V_{\!\rightarrow}+2\,w_{\leftarrow}>0$\,.
Considering a headwind, both summands are positive, which entails $\partial^2P/\partial V_{\!\rightarrow}^{\,2}>0$\,.
Considering a tailwind, the accepted condition, $V_{\!\rightarrow}>w_{\leftarrow}$\,, ensures that $3\,V_{\!\rightarrow}+2\,w_{\leftarrow}>0$\,, for $V_{\!\rightarrow}>0$\,.

Thus, under conditions stated in Proposition~\ref{prop:one-to-one}, the relation between power and speed is one-to-one and $P$ as a function of $V_{\!\rightarrow}$ is concave-up.
\end{proof}

This bijection means that\,---\,for steady rides along flats or uphill\,---\,the power generated by a cyclist and the ground speed of the bicycle-cyclist system are related by an invertible function, which is consistent with unique physical solutions obtained in Appendices~\ref{app:ConstraintWork} and \ref{app:ConstraintPower}.
Let us state two corollaries of Proposition~\ref{prop:one-to-one}.
\medskip
\begin{corollary}
Ceteris paribus, an increase of speed requires an increase of power, and an increase of power results in an increase of speed.
\end{corollary}
Also, as illustrated in Figure~\ref{fig:FigPowerSpeed}, the higher the value of $V_{\!\rightarrow}$\,, the more power is required for the same increase of speed.
\medskip
\begin{corollary}
\label{corr:2}
Ceteris paribus, a constant power results in a constant speed, an increase of power results in an instantaneous increase of speed, a decrease of power results in an instantaneous decrease of speed.
\end{corollary}
This corollary is illustrated by \citet[Figure~17]{BSSS}, for a velodrome.

To conclude this section, let us comment on the requirement for the ground speed to be greater than the tailwind.
If the ground speed is equal to the tailwind, there is no air resistance.
Since, for tailwinds, $w_{\leftarrow}$ is negative, it follows that the last summand in model~(\ref{eq:model}), which corresponds to the air resistance, is zero.
If the tailwind is greater than the ground speed, the wind does not oppose the movement but contributes to the propulsion.
To accommodate this effect within the model, it would be necessary to, at least, put the negative sign in front on the last summand.
However, ${\rm C_{d}A}$ refers to a frontal surface area; hence, the model requires the positive sign in front on the aforementioned summand.
\subsection{Fixed-wheel drivetrain}
\label{sub:Qualifier}
\begin{figure}
	\centering
	\includegraphics[scale=0.5]{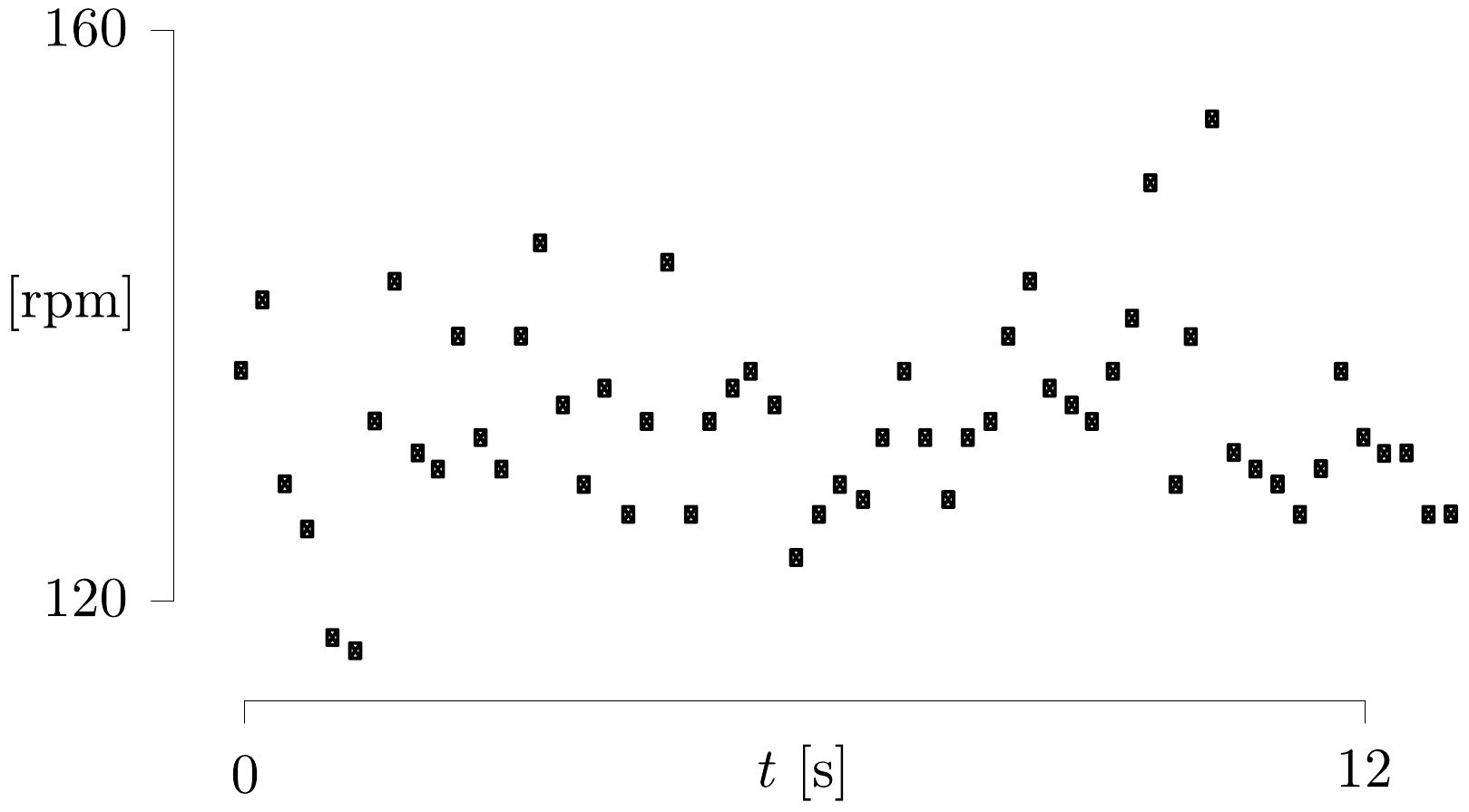}
	\caption{\small Measured cadence}
	\label{fig:FigCadence}
\end{figure}
Let us illustrate the difference between the power generated instantaneously by a cyclist and the measured power, for a fixed-wheel drivetrain.
The measured value of $P$ corresponds to the power generated instantaneously by a cyclist only if the pedal speed,~$v_{\circlearrowright}$\,, is an instantaneous consequence of the force applied to the pedals,~$f_{\circlearrowright}$\,, which is the case of a free-wheel drivetrain; if no force is applied to pedals, the pedals do not rotate, regardless of the bicycle speed.
Given a crank length, $v_{\circlearrowright}$ is obtained from measurements of cadence.
For a fixed-wheel drivetrain, for which there is a one-to-one relation between $v_{\circlearrowright}$ and the wheel speed, the momentum of the bicycle-cyclist system\,---\,which results in a pedal rotation, even without any force applied by a cyclist\,---\,might affect the value of~$v_{\circlearrowright}$\,.
\begin{figure}
	\centering
	\includegraphics[scale=0.5]{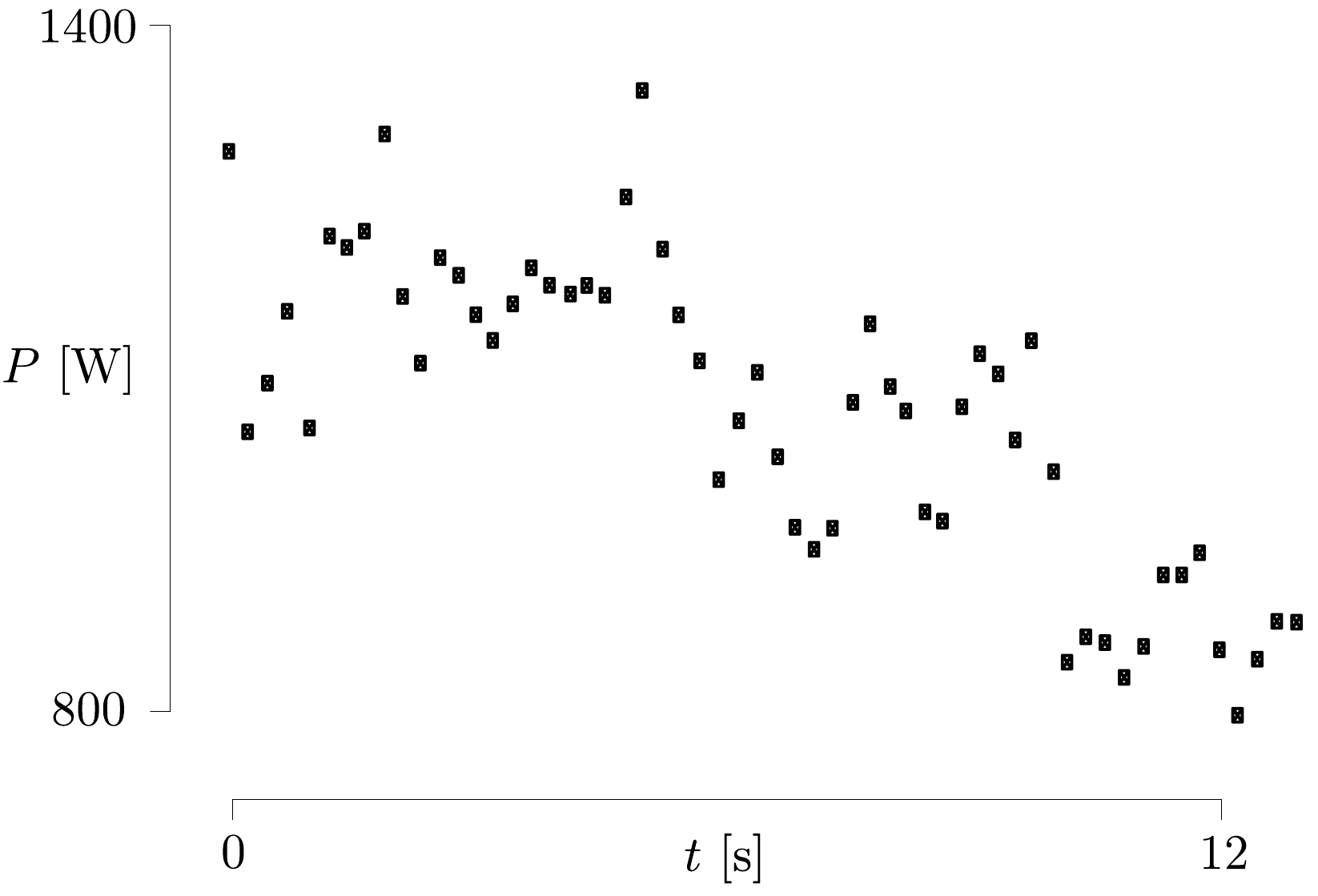}
	\caption{\small Measured power,~$P$}
	\label{fig:FigPowerFixed}
\end{figure}

To illustrate this issue, let us examine Figures~\ref{fig:FigCadence}, \ref{fig:FigPowerFixed} and \ref{fig:FigForce}, which correspond to the fixed-wheel measurements of, respectively, cadence, power and force from the second lap of a 1000-metre time trial on an indoor velodrome~(Mehdi Kordi, {\it pers.~comm.}, 2020).
Allowing for inaccuracies of measurements, we see that the first figure exhibits a steadiness of cadence,%
\footnote{Strictly speaking, unless a sufficient force is applied to the pedals, there is necessarily a decrease of cadence.
 Only hypothetically\,---\,with no internal or external resistance on the flats\,---\,the cadence would remain constant.}
 which entails the steadiness of~$v_{\circlearrowright}$ and the steadiness of the wheel speed.
The second figure exhibits a decrease of power, which\,---\,in view of the steadiness of~$v_{\circlearrowright}$ and in accordance with expression~\eqref{eq:formula}\,---\,is a consequence of the decrease of~$f_{\circlearrowright}$\,, shown in the third figure.
The behaviours observed in these figures are confirmed by the Augmented Dickey-Fuller test \citep{DickeyFuller}.
However, $F_{\!\leftarrow}$\,, in model~\eqref{eq:model}, need not decrease, since\,---\,in view of the steadiness of the wheel speed\,---\,$V_{\!\rightarrow}$ is approximately steady,%
\footnote{On a velodrome, $V_{\!\rightarrow}$ is equivalent to the wheel speed along the straights but not along the curves, due to the leaning of the bicycle-cyclist system.
Along the curves, and in general, $V_{\!\rightarrow}$ corresponds to the centre-of-mass speed~\citep{BSSS}.}
and so can be the power of the bicycle-cyclist system, for which we do not have direct measurements.
\begin{figure}
	\centering
	\includegraphics[scale=0.5]{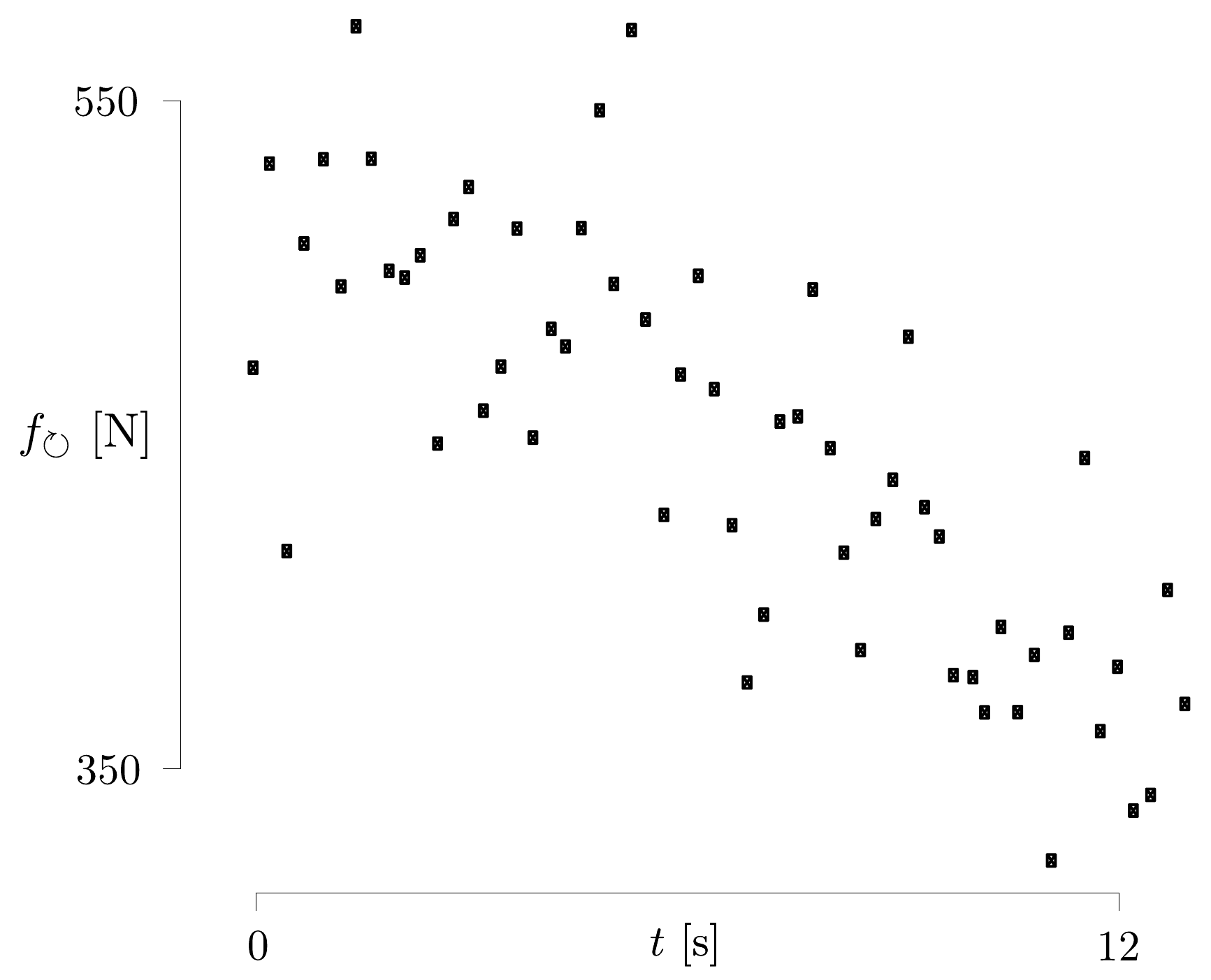}
	\caption{\small Measured force,~$f_{\circlearrowright}$}
	\label{fig:FigForce}
\end{figure}

In general, the measurements of power connected to the drivetrain and the power of the system itself are distinct from one another.
For instance, for a free-wheel drivetrain, on a downhill on which a cyclist does not pedal, $f_{\circlearrowright}\,v_{\circlearrowright}=0\neq F_{\!\leftarrow}\,V_{\!\rightarrow}>0$\,.%
\footnote{On a downhill, due to gravitation, the bicycle-cyclist system can accelerate even if the cyclist does not apply any force to the pedals.
On a flat, {\it ceteris paribus}, the system, albeit gradually, must slow down.}
Similarly, for a fixed-wheel drivetrain, the instantaneous measurements of power need not represent the power generated by a cyclist at these instances, since the rotation of the pedals might be partially due to the force exerted by the cyclist, at a given moment, and partially due to the momentum of the already moving bicycle-cyclist system.
The issue remains even if the sensors are not incorporated within the pedals but in the cranks, a bottom bracket or a rear hub.
In each case, the sensors are connected to the drivetrain.

Thus, for a fixed-wheel drivetrain, the power-meter measurements represent an instantaneous power generated by a cyclist if the power and cadence are in a dynamic equilibrium.
Such an equilibrium is reachable\,---\,following an initial acceleration\,---\,during a steady effort, as is the case of a $4000$\,-metre individual pursuit on a velodrome~\citep[e.g.,][Section~5]{BSSS} or the Hour Record, in contrast to accelerations followed by decelerations.
In the context of expression~\eqref{eq:formula}, a dynamic equilibrium means that changes of $v_{\circlearrowright}$ are immediate responses solely to changes in~$f_{\circlearrowright}$\,. 
\subsection{Air-density formula}
\label{sec:AirDens}
The air-density formula in model~\eqref{eq:model} is formulated from the ideal gas law with the assumption of hydrostatic equilibrium, which is a standard procedure outlined in works concerning atmospheric physics~(e.g., \citet[Chapter 2.2]{BohrenAlbrecht1998} or~\citet{Berberan-SantosEtAl2009}).
The experiment requires considering a thin slab of air of cross-sectional area~$A$ and thickness~$\Delta z$\,, whose total mass is the product of the mass density,~$\rho$\,, and volume,~$A\,\Delta z$\,, of the air contained therein.
The assumption of hydrostatic equilibrium requires the net sum of the forces acting on it amount to zero.
Thus, the so-called pressure-gradient force is directed upward\,---\,the upward pressure at the bottom of the slab, $p(z)$, is greater than the downward pressure at the top of the slab, $p(z+\Delta z)$\,---\,and is balanced by the downward weight of the slab, which results in
\begin{equation*}
	p(z) - p(z+\Delta z) - \rho\,A\,\Delta z\,g = 0\,,
\end{equation*}
where $g$ is the gravitational acceleration.
Dividing both sides by $A\,\Delta z$ and taking the limit of $\Delta z\to0$, we obtain the so-called hydrostatic equation,
\begin{equation}
	\label{eq:HySteq}
	\frac{{\rm d}p}{{\rm d}z} = -\rho\,g\,.
\end{equation}

Invoking the ideal gas law, $p\,V = n\,R\,T$\,, where $V$ is volume, $n$ is the number of moles, $R$ is the universal gas constant, $T$ is the temperature of the air in the slab, as well as considering the mass of the air,~$m$\,, with $n=m/M$\,, where $M$ is the molar mass of the air, and using $\rho = m/V$\,, we obtain
\begin{equation}
	\label{eq:rhoIGL}
	\rho = \frac{p\,M}{R\,T}\,.
\end{equation}
To solve equation~\eqref{eq:HySteq} using separation of variables, we substitute expression~\eqref{eq:rhoIGL} and divide both sizes by $p$ to obtain
\begin{equation*}
	\frac{{\rm d}p}{p} = -\frac{gM}{RT}\,{\rm d}z\,.
\end{equation*}
Setting the base of the slab at sea level, whereby $z_0=0\,{\rm m}$ and $p_0:=p(z_0)$, and integrating along the thickness of the slab, we have
\begin{equation*}
	\int\limits_{p_0}^{p(\Delta z)}\frac{{\rm d}p}{p} 
	= 
	-\int\limits_{0}^{\Delta z}\frac{gM}{RT}\,{\rm d}z\,,
\end{equation*}
whose solution is
\begin{equation*}
	\log p(\Delta z) - \log p_0 = -\frac{gM}{RT}\,\Delta z\,.
\end{equation*}
Taking the exponential of both sides and solving for $p(\Delta z)$, we have
\begin{equation*}
 	p(\Delta z) =  p_0\,\exp\left(-\frac{gM}{RT}\,\Delta z\right).
\end{equation*}
which is the barometric formula.
Defining altitude as $h:=\Delta z$ and substituting the barometric formula in expression~\eqref{eq:rhoIGL}, we have
\begin{align}
	\label{eq:rhoIGL}
	\rho = \frac{pM}{RT} = \frac{M}{RT}\left(p_0\exp\left(-\frac{gM}{RT}\,h\right)\right).
\end{align}
Upon substituting standard values~\citep[Table~1]{ICAO1993},
\begin{equation*}
	p_0 = 101325\,{\rm Pa},\quad
	g = 9.80665\,\frac{\rm m}{{\rm s}^2},\quad
	M = 0.02895442\,\frac{\rm kg}{\rm mol},\quad
	R = 8.3132\,\frac{\rm J}{{\rm mol}\,{\rm K}},\quad
	T = 288.15\,{\rm K} = 15^\circ{\rm C},
\end{equation*}
expression~\eqref{eq:rhoIGL} becomes
\begin{equation*}
	\rho(h) 
	= 
	\left(1.225\,\frac{\rm kg}{{\rm m}^3}\right)
	\exp\left(-\frac{0.0001186}{\rm m}\,h\right)
	,
\end{equation*}
which is air-density formula~\eqref{eq:DenAlt} used in model~\eqref{eq:model}.
\subsection{Air-resistance coefficient}
\label{sec:AirResCoeff}
To formulate the air-resistance force in model~\eqref{eq:model}, we assume that it is proportional to the frontal area,~$A$\,, and to the pressure,~$p$\,, exerted by air on this area,~$F_a\propto p\,A$\,, where $p=\tfrac{1}{2}\rho\,V^2$ has a form of kinetic energy and $V=V_{\!\rightarrow}+w_\leftarrow$ is the relative speed of a cyclist with respect to the air; $p$ is the energy density per unit volume.
We write this proportionality as
\begin{equation*}
F_a=\tfrac{1}{2}\,{\rm C_d}\,A\,\rho\,V^2\,,
\end{equation*}
where $\rm C_d$ is a proportionality constant, which is referred to as the drag coefficient.

A more involved justification for the form of the air-resistance force in model~\eqref{eq:model} is based on dimensional analysis \citep[e.g.,][Chapter~3]{Birkhoff}.
We consider the air-resistance force, which is a dependent variable, as an argument of a function, together with the independent variables, to write
\begin{equation*}
f(F_a,V,\rho,A,\nu)=0\,; 	
\end{equation*}
herein, $\nu$ is the viscosity coefficient.
Since this function is zero in any system of units, it is possible to express it only in terms of dimensionless groups.

According to the Buckingham theorem \citep[e.g.,][Chapter~3, Section~4]{Birkhoff}\,---\,since there are five variables and three physical dimensions, namely, mass, time and length\,---\,we can express the arguments of~$f$ in terms of two dimensionless groups.
There are many possibilities of such groups, all of which lead to equivalent results.
A common choice for the two groups is
\begin{equation*}
	\frac{F_a}{\tfrac{1}{2}\,\rho\,A\,V^2}\,,
\end{equation*}
which is referred to as the drag coefficient, and
\begin{equation*}
\frac{V\,\sqrt{A}}{\nu }\,,
\end{equation*}
which is referred to as the Reynolds number.
Thus, treating physical dimensions as algebraic objects, we can reduce a function of five variables into a function of two variables,
\begin{equation*}
g\left(\frac{F_a}{\tfrac{1}{2}\,\rho\,A\,V^2}\,,\,\frac{V\,\sqrt{A}}{\nu }\right)\,=\,0\,,
\end{equation*}
which we write as
\begin{equation*}
\frac{F_a}{\tfrac{1}{2}\,\rho\,A\,V^2}=h\left(\frac{V\,\sqrt{A}}{\nu }\right)\,,
\end{equation*}
where the only unknown is~$F_a$\,, and where $h$ is a function of the Reynolds number.
Denoting the right-hand side by $\rm C_d$\,, we obtain
\begin{equation*}
F_a=\tfrac{1}{2}{\rm C_d}\,A\,\rho\,V^2\,,
\end{equation*}
as expected.
In view of this derivation, $\rm C_d$ is not a constant; it is a function of the Reynolds number.
In our study, however\,---\,within a limited range of speeds\,---\,$\rm C_d$ is treated as a constant.
Furthermore, since $A$ is difficult to estimate, we include it within this constant, and consider their product,~$\rm C_dA$\,.

\subsection{Rotation effects: air resistance}
\label{sec:AirResRot}
To include the effect of air resistance of rotating wheels in model~\eqref{eq:model}, another summand would need to be introduced in the numerator,
\begin{equation*}
\tfrac{1}{2}{\rm C_w}\pi r^2\rho\,(V_{\!\rightarrow}+w_{\leftarrow})^2\,,	
\end{equation*}
where $r$ is the wheel radius.
Such a summand is formulated by invoking dimensional analysis in a manner analogous to the one presented in Section~\ref{sec:AirResCoeff}.

To combine rotational air resistance with the translational one, we use $v=\omega\,r$\,, where $v$ is the circumferential speed and $\omega$ is the angular speed, and the fact that\,---\,as discussed in Section~\ref{sec:RotEff} and illustrated in Figure~\ref{fig:FigNoSlip}\,---\,the circumferential speed, under the assumption of rolling without slipping, is the same as the ground speed of the bicycle,~$V_{\!\rightarrow}$\,.

Considering two standard wheels, we write
	\begin{equation}
		\label{eq:model2}
		P
		=
		\frac{
			mg\sin\theta 
			+
			(m+2\,m_{\rm w})\,a 
			+
			{\rm C_{rr}}m g\cos\theta 
			+
			\tfrac{1}{2}\,\rho
			\,(
				2\,{\rm C_w}\overbrace{\pi r^2}^{A_\circ}
				+
				{\rm C_{d}A_f}
			)
			\left(V_{\!\rightarrow} + w_{\leftarrow}\right)^2
		}{
			1-\lambda
		}
		V_{\!\rightarrow}
		\,;	
	\end{equation}
herein, in contrast to model~\eqref{eq:model} and as discussed in Section~\ref{sec:RotEff}, the change of speed, expressed by the second summand, contains effects of the moment of inertia due to rolling wheels.
The air resistance, expressed by the fourth summand, distinguishes between the air resistance due to translation of a bicycle from the air resistance due to its rolling wheels; $\rm A_f$ is the entire frontal area and $\rm A_\circ$ is the wheel side area.%
\footnote{Using model~\eqref{eq:model2} and the implicit function theorem, ${\partial\,{\rm C_dA_f}/\partial\,{\rm C_wA_\circ}=-2}$\,, which is indicative of the behaviour of a model; there is no physical relation between $\rm C_dA_f$ and $\rm C_wA_\circ$\,.
We expect, $\rm C_wA_o \ll C_dA_f$\,; however, the optimization programs treat them as two adjustable parameters of equal importance.}
An examination of the effect of two different wheels requires the introduction of two coefficients, one for each wheel.

In this study, the quality of available information renders the extraction of values for the resistance coefficients difficult.
Even though the data obtained from the power meter exhibit high accuracy, the data based on the GPS measurements introduce the uncertainty that renders an accurate extraction of even three parameters a numerical challenge.
Extraction of four or five parameters requires more accurate data.

In the meantime, however, we can consider forward models to gain an insight into the effect of a disk wheel.
Following model~\eqref{eq:model2}\,---\,for a steady ride,~$a=0$\,, under windless conditions,~$w_{\leftarrow}=0$\,, on a flat course,~$\theta=0$\,---\,we write the required powers as
\begin{equation*}
	P_{\rm n}=\dfrac{{\rm C_{rr}}mg+\tfrac{1}{2}\,\rho\left(2\,{\rm C_{w_n}A_w}+{\rm C_{d}A_f}\,\right)V_{\!\rightarrow}{}^2}{1-\lambda}V_{\!\rightarrow}	
	\quad{\rm and}\quad
	P_{\rm d}=\dfrac{{\rm C_{rr}}mg+\tfrac{1}{2}\,\rho\left(({\rm C_{w_n}}+{\rm C_{w_d}}){\rm A_w}+{\rm C_{d}A_f}\,\right)V_{\!\rightarrow}{}^2}{1-\lambda}V_{\!\rightarrow}\,,	
	\end{equation*}
where we distinguish between the drag coefficients of a normal wheel and a disk wheel.

The difference in required power is
\begin{equation*}
	\Delta P
	=
	\dfrac{{\rm C_{w_n}}-{\rm C_{w_d}}}{2\left(1-\lambda\right)}\,{\rm A_w}\,\rho\,{V_{\!\rightarrow}}^{3}\,.
\end{equation*}
Letting \mbox{${\rm C_{w_n}}\approx 0.05$} and \mbox{${\rm C_{w_d}}\approx 0.035$}~\citep{GreenwellEtAl1995} means that, for a standard wheel, \mbox{${\rm C_wA_\circ}\approx 0.015\,{\rm m}^2$}\,, and for a disk wheel,~\mbox{${\rm C_wA_\circ}\approx 0.01\,{\rm m}^2$}\,.
Both values are significantly smaller than $\rm C_dA_f$\,, as expected.
Letting $r=0.31\,{\rm m}$\,, $\rho=1.204\,{\rm kg/m}^3$\,, $\lambda=0.03574$\,, we obtain $\Delta P\approx0.0028\,V_{\!\rightarrow}{}^3$\,.
For $\overline V_{\!\rightarrow}=10.51\,{\rm m/s}$\,, we have $\Delta P\approx 3.3\,{\rm W}$\,.
Thus, for the present study, in the neighbourhood of $\overline P=258.8\,{\rm W}$\,, the replacement of a regular wheel by a disk wheel results in the decrease of required power of $\Delta P/\overline P\approx 1.3\%$\,, to maintain the same speed.
\subsection{Rotation effects: moment of inertia}
\label{sec:RotEff}
To include the effect of rotation upon change of speed in model~\eqref{eq:model}, we would need to consider the moment of inertia, which is $m\,r^2$\,, for a thin circular loop, and $m\,r^2/2$\,, for a solid disk, where $r$ stands for their radii.
	Relating the angular change in speed to the circumferential one\,---\,by a temporal derivative of $v=\omega r$\,, where $v$ is the circumferential speed and $\omega$ is the angular speed\,---\,the magnitudes of the corresponding rotational force are $Fr=mra$ and  $Fr=mra/2$\,, respectively; hence, $F=ma$ and $F=ma/2$ are the corresponding linear forces.
	
	In the preceding paragraph, $v=\omega r$ is the circumferential speed.
	To show that it is equal to the ground speed of the bicycle,~$V_{\!\rightarrow}$\,, let us consider the point of contact of the wheel and the road.
	The ground speed of that point is the sum of the circumferential speed of the wheel, at that point, and the speed of the bicycle.
	Since\,---\,under assumption of no slipping\,---\,the ground speed of that point is zero and the other two speeds refer to velocities in the opposite directions, we have $v=V_{\!\rightarrow}$\,.
	Thus, as illustrated in Figure~\ref{fig:FigNoSlip}, the circumferential speed of the wheel is the same as the ground speed of the bicycle; the same is true for the change of speed.
\begin{figure}
\centering
\includegraphics[scale=0.7]{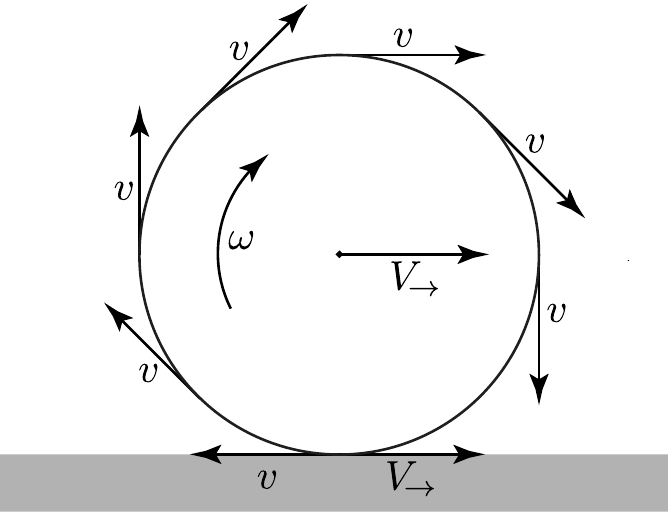} 
\caption{\small Rolling without slipping: angular speed,~$\omega$\,, circumferential speed,~$v=\omega\,r$\,, where $r$ is radius, and bicycle speed,~$V_{\!\rightarrow}=v$}
\label{fig:FigNoSlip}
\end{figure}
	
To consider a bicycle with one standard wheel and one disk wheel, we denote the mass of the former by~$m_{\rm w}$\,, and the mass of the latter by~$m_{\rm d}$\,.
Thus, the second summand in the numerator of model~\eqref{eq:model}, which is a linear force, becomes $(m+m_{\rm w}+m_{\rm d}/2)\,a$\,.
	
Note that $m$ contains both $m_{\rm w}$ and $m_{\rm d}$ to account for the translational and rotational effects; the former depends on the total mass and the latter on the mass of the wheels only.
For a standard wheel, $m_{\rm w}\approx 0.8\,{\rm kg}$\,, and for a disk wheel,~$m_{\rm d}\approx 1.3\,{\rm kg}$\,; for both~$r=0.31\,{\rm m}$\,.

\section{Discussion}
One of the main purposes of this article is to estimate the effects of air, rolling and drivetrain resistance.
To estimate the corresponding parameters, we restrict the study to ranges of speeds that are consistent with time trialing on either a flat course or a course with a steady inclination.
In either case, a rider pedals without interruptions and in the same position.
Also, there are no tight corners that would require braking. 
These restrictions allow us\,---\,for either case\,---\,to treat $\rm C_{d}A$\,, $\rm C_{rr}$ and $\lambda$ as constants, even though, in general, they are functions of speed.

The force of the air resistance, which is expressed by the fourth summand of model~\eqref{eq:model}, is proportional to the square of speed and, hence\,---\,for flat courses\,---\,it becomes dominant.
For climbs, in contrast, the expression is dominated by the first summand, which accounts for the force against gravity.
As discussed in Section~\ref{sec:RatesOfChange}, the effect of the air resistance can be quantified by relations between $\rm C_{d}A$ and $V_{\!\rightarrow}$\,.
The effect of gravity can be quantified by relations among $m$\,, $\theta$ and $V_{\!\rightarrow}$\,.
For instance, for any cyclist\,---\,with a given power output\,---\,the dependence of speed on mass is greater on a climb than on a flat course, and this dependence can be quantified, using $\partial_mV_{\!\rightarrow}$ and the implicit function theorem, discussed in Section~\ref{sec:PhysicalInferences}.

In estimating the effects of the air, rolling and drivetrain resistance, we recognize that the right-hand side of model~\eqref{eq:model}, which is a forward problem, invokes $\rm C_{d}A$\,, $\rm C_{rr}$ and $\lambda$ with their independent physical meanings.
The misfit minimization of expression~\eqref{eq:misfit}, however\,---\,for which the left-hand side is the measured value and the right-hand side is the retrodiction of a model\,---\,treats $\rm C_{d}A$\,, $\rm C_{rr}$ and $\lambda$ as adjustable parameters.
Relations between the rates of change of any two quantities, in expression~\eqref{eq:misfit}\,---\,resulting from the implicit function theorem\,---\,are insightful for interpreting the behaviour of the model.
For instance, $\rm C_{d}A$ and $\rm C_{rr}$\,, which\,---\,as physical quantities\,---\,are independent of one another, become inversely proportional to one another, as discussed in Section~\ref{sec:ModelInferences}.

Maintaining the physical meaning of $\rm C_{d}A$\,, $\rm C_{rr}$ and $\lambda$ is a challenge.
To extract~$\lambda$\,, one might consider the placement of the power meter and its effect on drivetrain loss.
As indicated by \citet{Chung2012}, a power meter in the rear hub provides readings absent of drivetrain loss.
This is not the case for power measured elsewhere.
To accommodate this issue, \citet{MartinEtAl1998} suggests a fixed wattage or percentage loss.

Setting, {\it a priori}\,, $\lambda$ to be a fixed value increases the stability of extracting the remaining two coefficients in model~\eqref{eq:model}.
Otherwise, a division of the entire expression is a scaling that contributes to nonuniqueness.
Even so, ``prying apart'' \mbox{$\rm C_dA$ and~$\rm C_{rr}$} remains a challenge, for which \citeauthor{Chung2012} suggests testing the same tires and tubes on the same roads on the same day at the same pressure so that $\rm C_{rr}$ is a constant, and then concentrate on estimating changes in $\rm C_dA$\,.

Another purpose of this article is to use the results derived herein as sources of information for optimizing the performance in a time trial under a variety of conditions, such as the strategy of the distribution of effort over the hilly and flat portions or headwind and tailwind sections.
For instance, examining $\partial_{w_{\leftarrow}}V_{\!\rightarrow}$\,, for a flat course, we could quantify a time-trial adage of pushing against the headwind and recovering with the tailwind, to diminish the overall time, under a constraint of cyclist's capacity; such a conclusion is illustrated in Appendix~\ref{app:LagrangeMultipliers}.
A further insight into this statement is provided by the following example.

Let us consider a five-kilometre section against the headwind,~$w_{\leftarrow}=5\,{\rm m/s}$\,, and, following a turnaround, the same five-kilometre section with the tailwind, $w_{\leftarrow}=-5\,{\rm m/s}$\,.
If we keep a constant power,~$P=258.8\,{\rm W}$\,, and use model~\eqref{eq:model} to find the corresponding speed, we achieve the total time of $946$~seconds, for ten kilometres, with the upwind speed of $V_{\!\rightarrow}=8.273\,{\rm m/s}$ and the downwind speed of $V_{\!\rightarrow}=14.63\,{\rm m/s}$\,.
If we maintain the same average power, over ten kilometres, but increase the power on the upwind section by $10\%$ and decrease the power on the downwind section by $10\%$\,, we reduce the total time by about $16$~seconds, with the upwind speed of $V_{\!\rightarrow}=8.647\,{\rm m/s}$ and the downwind speed of $V_{\!\rightarrow}=13.78\,{\rm m/s}$\,.
For reliable results\,---\,in view of Figure~\ref{fig:FigPowerSpeed} and the linear approximation within a neighbourhood of the average speed for which the flat-course model is established\,---\,one should not consider excessive increases or decreases of speed or power.
To conclude this example, let us consider the case of keeping a constant speed,~$V_{\!\rightarrow}=11.43\,{\rm m/s}$\,, which is the average for the latter scenario, for ten kilometres.
Such a strategy requires the power for  the upwind section to be $P=531.6\,{\rm W}\gg 258.8\,{\rm W}$\,.
Thus, even though we should push harder against the wind than with the wind, we should not try to keep the same speed for both the upwind and downwind sections.
This conclusion is consistent with the partial-derivative values of Table~\ref{table:PhysicalRates}.

This conclusion is\,---\,only in part\,---\,consistent with a ``Rule of Thumb'' of \citet{Anton2013}, which amounts to expending some extra energy when riding against the wind and conserving some energy when riding with the wind.
A quantification of this, and another, rule of thumb of \citet{Anton2013} is presented in Appendix~\ref{app:LagrangeMultipliers}, where we question their generality.

\section{Summary}
In this article, we present a mathematical model that accounts for the power required to overcome the forces opposing the translational motion of a bicycle-cyclist system.
The model permits the estimation of physical parameters, such as air, rolling and drivetrain resistances.
Under various conditions, there are different relations between the rates of change of the model parameters in question. As a consequence of the implicit function theorem, relations between the rates of change of all quantities that are included in a model are explicitly stated, and each relation can be evaluated for given conditions.

Furthermore, the derived expressions allow us to interpret the obtained measurements in a quantitative manner, since the values of these expressions entail concrete issues to be addressed for a given bicycle course.
The reliability of information\,---\,which depends on the accuracy of measurements and the empirical adequacy of a model\,---\,is quantified by a misfit and by standard deviations of model parameters.
\section*{Acknowledgements}
We wish to acknowledge Len Bos, Yves Rogister and Rapha\"el Slawinski, for fruitful discussions, Mehdi Kordi for providing the measurements used in Section~\ref{sub:Qualifier}, David Dalton, for his scientific editing and proofreading, Elena Patarini, for her graphic support, and Roberto Lauciello, for his artistic contribution of famous mathematicians riding modern bicycles.
Furthermore, we wish to acknowledge Favero Electronics, including their Project Manager, Francesco Sirio Basilico, and R\&D engineer, Renzo Pozzobon, for inspiring this study by their technological advances and for supporting this work by insightful discussions and providing us with their latest model of Assioma Duo power meters.
Map data copyrighted OpenStreetMap contributors and available from \href{https://www.openstreetmap.org}{\tt https://www.openstreetmap.org}.
\bibliographystyle{apa}
\bibliography{DSSbici.bib}
\begin{appendix}
\section{On effects of averaging pedal speed per revolution for power calculations}
\label{app:AvgPedSpeed}
\setcounter{equation}{0}
\setcounter{figure}{0}
\renewcommand{\theequation}{\Alph{section}.\arabic{equation}}
\renewcommand{\thefigure}{\Alph{section}\arabic{figure}}
\subsection{Preliminary remarks}
Since, in expression~\eqref{eq:formula}, $v_{\circlearrowright}$ is proportional to the cadence, it is common to simplify circumferential speed measurements by considering only the cadence.
This means that\,---\,instead of measuring the speed instantaneously along the revolution\,---\,the measurement is performed only once per revolution and the resulting average is used in subsequent calculations.
In this appendix, we examine the effects of such a simplification.

Referring to an expression equivalent to expression~\eqref{eq:formula}\,---\,that invokes torque and angular velocity instead of the force and circumferential speed\,---\,\citet{Favero2018} state that
\begin{quote}
	[t]he torque/force value is usually measured many times during each rotation, while the angular velocity variation is commonly neglected, considering only its average value for each revolution.
	$[\,\ldots\,]$
	Favero Electronics, to ensure the maximum accuracy of its power meters in all pedaling conditions, decided to research to what extent the variation of angular velocity during a rotation affects the power calculation.	
\end{quote}
To examine the effect of including speed variation during a revolution, let us consider the following formulation to gain  analytical insights into the empirical results obtained by~\citet{Favero2018}.
\subsection{Formulation}
Consider a pedal whose revolution takes one second; hence, its circumferential speed is
\begin{equation}
	\label{eq:speed}
	v_{\circlearrowright}(\theta)
	=
	v_{0}\left(1+a\cos(2\theta)\right)
	\,,\qquad
	\theta\in(0,2\pi]
	\,,
\end{equation}
where $v_{0}=2\pi r/1$\,, $r$ is the crank length and $\theta$ is the angle.
Expression~(\ref{eq:speed}) is illustrated in Figure~\ref{fig:FigSpeed}.
\begin{figure}
	\centering
	\includegraphics[scale=0.5]{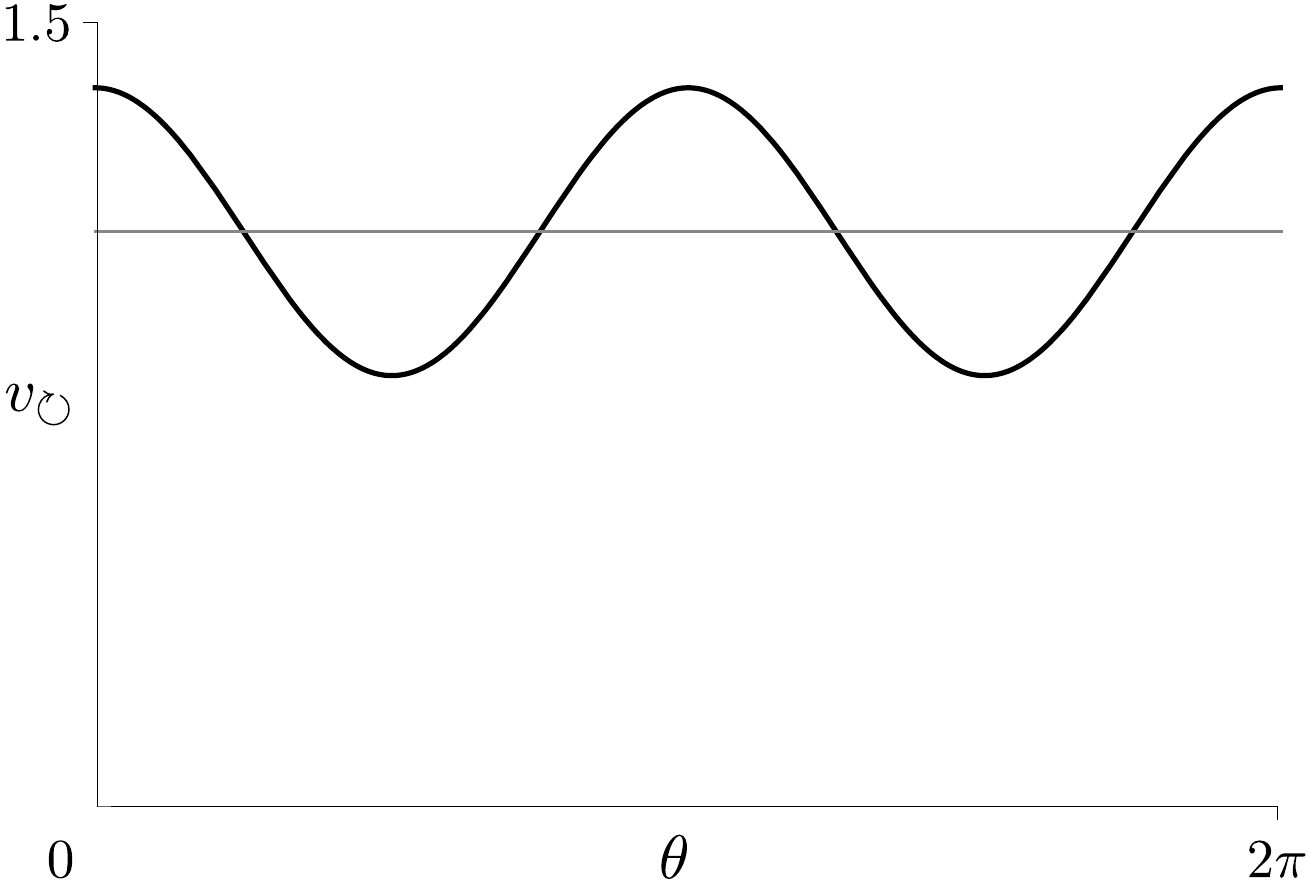}
	\caption{\small
		Circumferential speeds corresponding to expression~(\ref{eq:speed}): $a=0.25$\,, $r=0.175\,{\rm m}$ and $a=0$\,, $r=0.175\,{\rm m}$\,;  the former shown in black and the latter in grey
	}
	\label{fig:FigSpeed}
\end{figure}

Assume the magnitude of the tangential component of force applied to both pedals during this revolution to be
\begin{equation}
	\label{eq:force}
	f_{\circlearrowright}(\theta)
	=
	f_{0}\left(1+b\cos(2(\theta+c))\right)
	\,,\qquad
	\theta\in(0,2\pi]\,,	
\end{equation}
where $f_{0}$ is a constant and $c$ is an angular shift between $v_{\circlearrowright}$ and $f_{\circlearrowright}$\,, which is a constant whose units are radians.
Expression~(\ref{eq:force}) is illustrated in Figure~\ref{fig:FigFaveroForce}.
\begin{figure}
	\centering
	\includegraphics[scale=0.5]{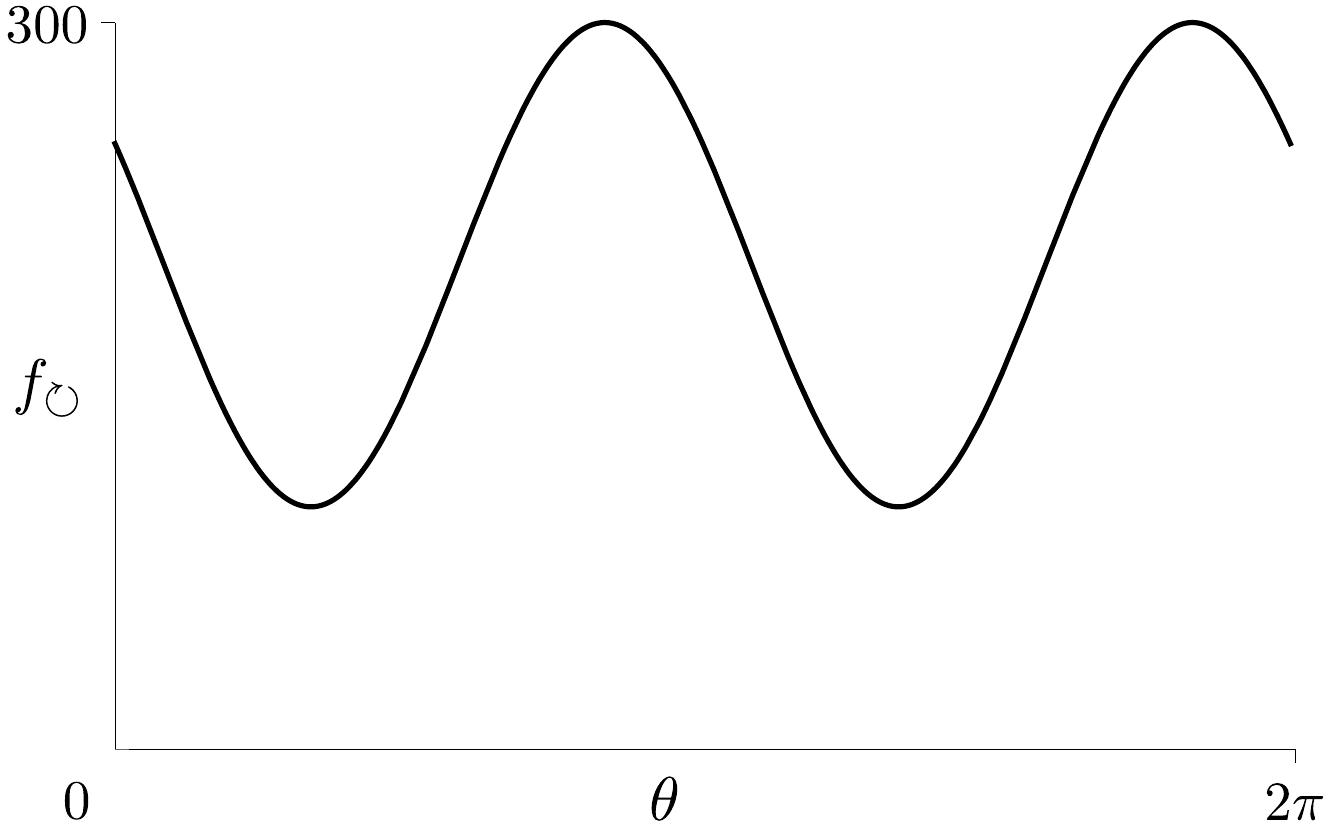}
	\caption{\small
		Applied force corresponding to expression~(\ref{eq:force}): $b=0.5$\,, $c=\pi/6$\,, $f_{0}=200\,{\rm N}$
	}
	\label{fig:FigFaveroForce}
\end{figure}

In accordance with expression~(\ref{eq:formula}), the instantaneous power, at $\theta$\,, is $P_{\circlearrowright}(\theta)=f_{\circlearrowright}(\theta)\,v_{\circlearrowright}(\theta)$\,, and the average power over the revolution is
\begin{equation}
	\label{eq:InsPow}
	\overline{P}_{\circlearrowright}
	=
	\dfrac{1}{2\pi}
	\int\limits_{0}^{2\pi}
	P_{\circlearrowright}(\theta)
	{\,\rm d}\theta
	=
	\left(2+a\,b\cos(2\,c)\right)\pi r f_{0}\,.
\end{equation}
If we consider the average value of speed,
\begin{equation}
	\label{eq:velocity}
	\overline{v}_{\circlearrowright}
	=
	\dfrac{1}{2\pi}
	\int\limits_{0}^{2\pi}
	v_{\circlearrowright}(\theta)
	{\,\rm d}\theta
	=
	2\pi r
	\,,	
\end{equation}
then,
\begin{equation}
	\label{eq:AvePow}
	\overline{P}_{\circlearrowright}
	=
	\dfrac{1}{2\pi}
	\int\limits_{0}^{2\pi}
	f_{\circlearrowright}(\theta)\,\overline{v}_{\circlearrowright}
	{\,\rm d}\theta
	=
	2\pi r f_{0}
	\,,	
\end{equation}
over one revolution;  $a$\,, $b$ and $c$ have no effect on $\overline{P}_{\circlearrowright}$\,.
Examining expressions~(\ref{eq:InsPow}) and (\ref{eq:AvePow}), we see that the former reduces to the latter if $a\,b=0$ or if $c=\pi/4$ or $c=3\pi/4$\,.
Otherwise, the power over a revolution\,---\,based on the instantaneous speed\,---\,is different from the power based on the speed averaged for each revolution.
One might note that expression~(\ref{eq:AvePow}) can be also obtained as the product of expression~(\ref{eq:velocity}) and
\begin{equation}
	\label{eq:BarF}
	\overline{f}_{\circlearrowright}
	=
	\frac{1}{2\pi}
	\int\limits_{0}^{2\pi}
	f_{0}\left(1+b\cos(2(\theta+c)\right)
	{\rm d}\theta
	=
	f_{0}
	\,,
\end{equation}
which is the average force per revolution that results from expression~(\ref{eq:force})\,.	
\begin{figure}
	\centering
	\includegraphics[scale=0.5]{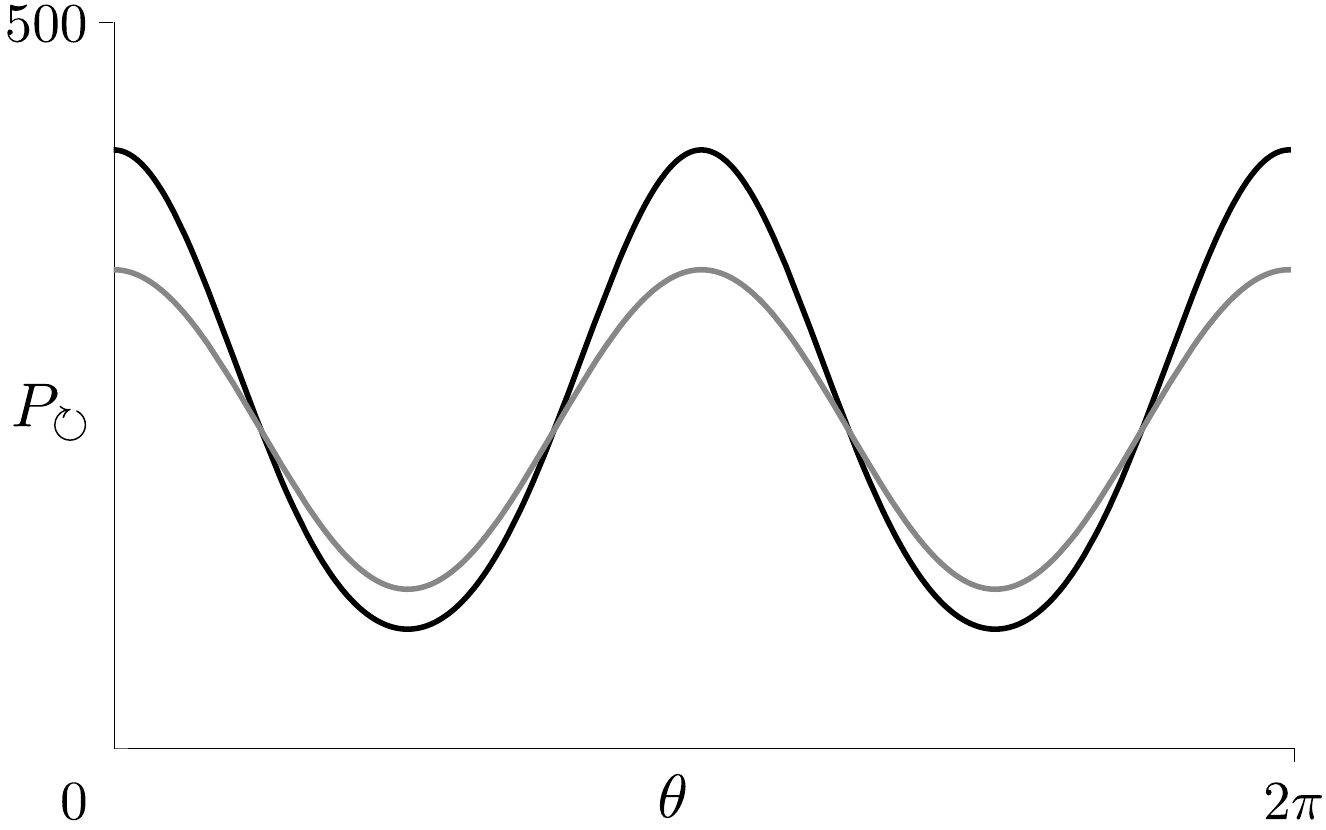}
	\caption{\small
		Instantaneous power corresponding to expressions~(\ref{eq:InsPow}) and (\ref{eq:AvePow}); the former shown in black and the latter in grey, with averages of $227\,{\rm W}$ and $220\,{\rm W}$\,, respectively
	}
	\label{fig:FigPowerRot}
\end{figure}

The integrands of expressions~(\ref{eq:InsPow}), with $c=0$\,, and (\ref{eq:AvePow}) are illustrated in Figure~\ref{fig:FigPowerRot}.
Therein, following expressions~(\ref{eq:speed}) and (\ref{eq:force}), the integrand of expression~(\ref{eq:InsPow})  is
\begin{equation*}
	P_{\circlearrowright}(\theta)
	=
	f_{\circlearrowright}(\theta)\,
	v_{\circlearrowright}(\theta)
	=
	f_{0}\,(1+b\cos(2(\theta+c)))\,v_{0}(1+a\cos(2\theta))
	\,.
\end{equation*}
Invoking trigonometric identities and rearranging, we write it as
\begin{equation*}
	P_{\circlearrowright}(\theta)
	=
	f_{0}\,v_{0}
	\left(
		\underbrace{1+\dfrac{ab}{2}\cos(2c)}_{\rm constant}
		+
		\underbrace{
			\dfrac{}{}a\cos(2\theta) + b\cos(2(\theta+c))
		}_{\rm double\,frequency}
		+
		\underbrace{
			\dfrac{ab}{2}\cos(4\theta+2c)
		}_{\rm quadruple\,frequency}\,
	\right)
	\,.
\end{equation*}
However, the effect of the third summand is small enough not to appear in Figure~\ref{fig:FigPowerRot}.
For instance, if we let $c=0$\,, the double-frequency term becomes $(a+b)\cos(2\theta)$ and the quadruple frequency term becomes $\tfrac{ab}{2}\cos(4\theta)$\,.
If $a<1$ and $b<1$\,, the amplitude of the third summand is much smaller, and the appearance of the plot is dominated by the double-frequency term.
\subsection{Numerical examples}
If we let $a= 0.25$\,, $b = 0.5$\,, $c = 0$\,, $f_{0} = 200\,{\rm N}$\,, $v_{0}=2\pi r$ and $r=0.175\,{\rm m}$\,, expression~(\ref{eq:InsPow}) results in $234\,{\rm W}$, as the average power per revolution, and expression~(\ref{eq:AvePow}) in $220\,{\rm W}$.
The approach that neglects speed variations during the revolution can also overestimate the average power.
If we let $c=\pi/2$\,, expression~(\ref{eq:InsPow}) results in $\overline{P}=206\,{\rm W}$ and expression~(\ref{eq:AvePow}) remains unchanged.
These results are illustrated in Figure~\ref{fig:FigAverage}.
\begin{figure}
	\centering
	\includegraphics[scale=0.5]{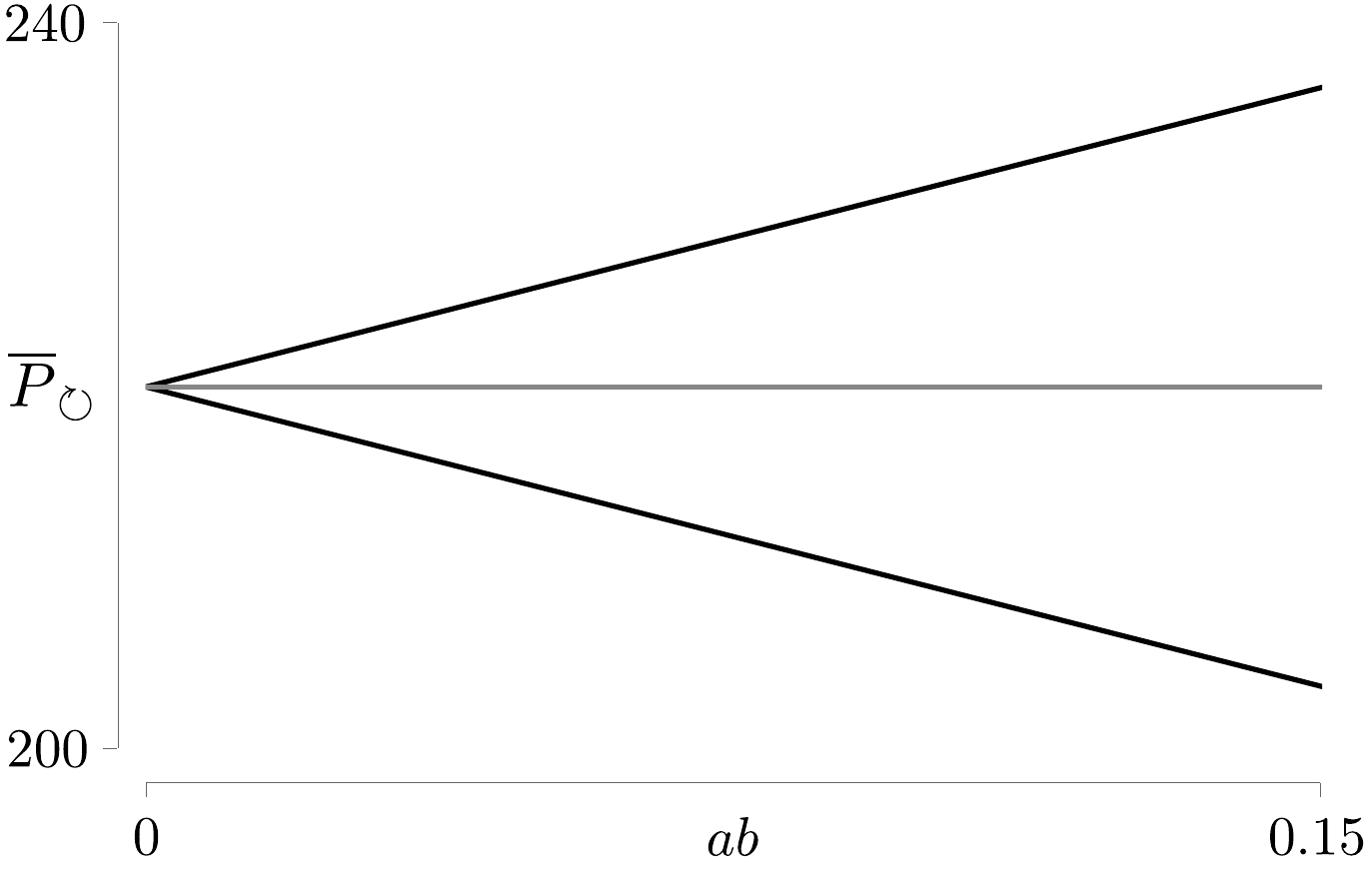}
	\caption{\small
		Average power corresponding to expressions~(\ref{eq:InsPow}) and (\ref{eq:AvePow}); the former shown in black and the latter in grey; the former depends on $ab$\,, the latter does not; the former depends on $c$\,, the latter does not; $c=0$\,, for the increasing line, $c=\pi/2$ or $c=3\pi/2$\,, for the decreasing line
	}
	\label{fig:FigAverage}
\end{figure}

Expressions~(\ref{eq:speed})--(\ref{eq:BarF}) refer to a single revolution.
Hence, the values resulting from expressions~(\ref{eq:InsPow}) and (\ref{eq:AvePow}) remain the same, regardless of cadence; they are averages over one rotation.

If the pedaling is smoother, as one might expect for higher cadences, the values of $a$ and $b$ become smaller.
Since these values are smaller than unity and appear as a product, expression~(\ref{eq:InsPow}) might approach expression~(\ref{eq:AvePow}).
If we let $a=0.1$\,, $b=0.3$\,, $c=0$\,, $f_{0}=200\,{\rm N}$ and $r=0.175\,{\rm m}$\,, expression~(\ref{eq:InsPow}) results in $\overline{P}=223\,{\rm W}$\,; the value of expression~(\ref{eq:AvePow}) remains unchanged.

Furthermore, for a single revolution, there is a unique pair of force and speed that results in a power given by expressions~(\ref{eq:InsPow}) and (\ref{eq:AvePow}).
However\,---\,for a given time interval and various cadences\,---\,there are many pairs of force and speed that result in the same value of power.

\begin{figure}
	\centering
	\includegraphics[scale=0.5]{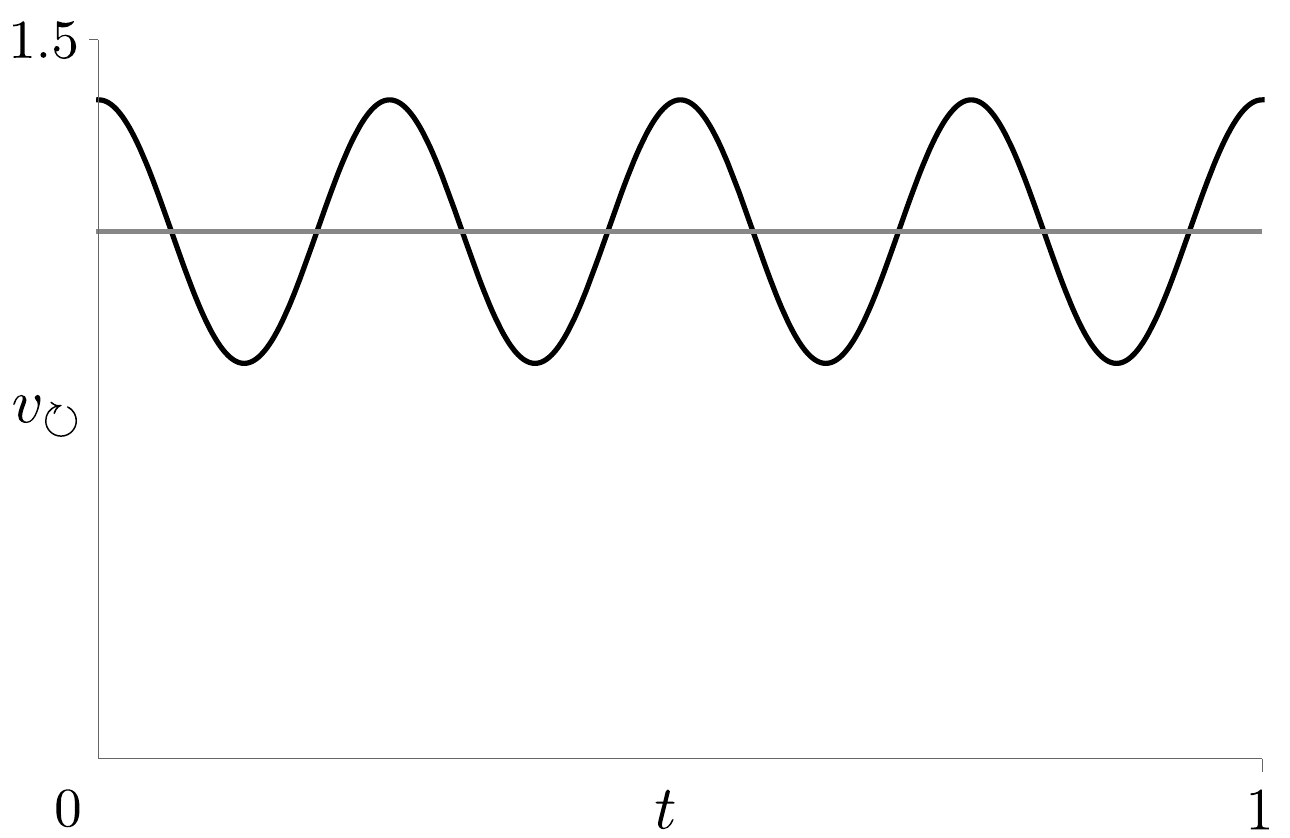}
	\caption{\small
		Circumferential speeds corresponding to expression~(\ref{eq:speed2}): $a=0.25$\,, $r=0.175\,{\rm m}$ and $a=0$\,, $r=0.175\,{\rm m}$\,;  the former shown in black and the latter in grey
	}
	\label{fig:FigRenzoOne}
\end{figure}
\begin{figure}
	\centering
	\includegraphics[scale=0.5]{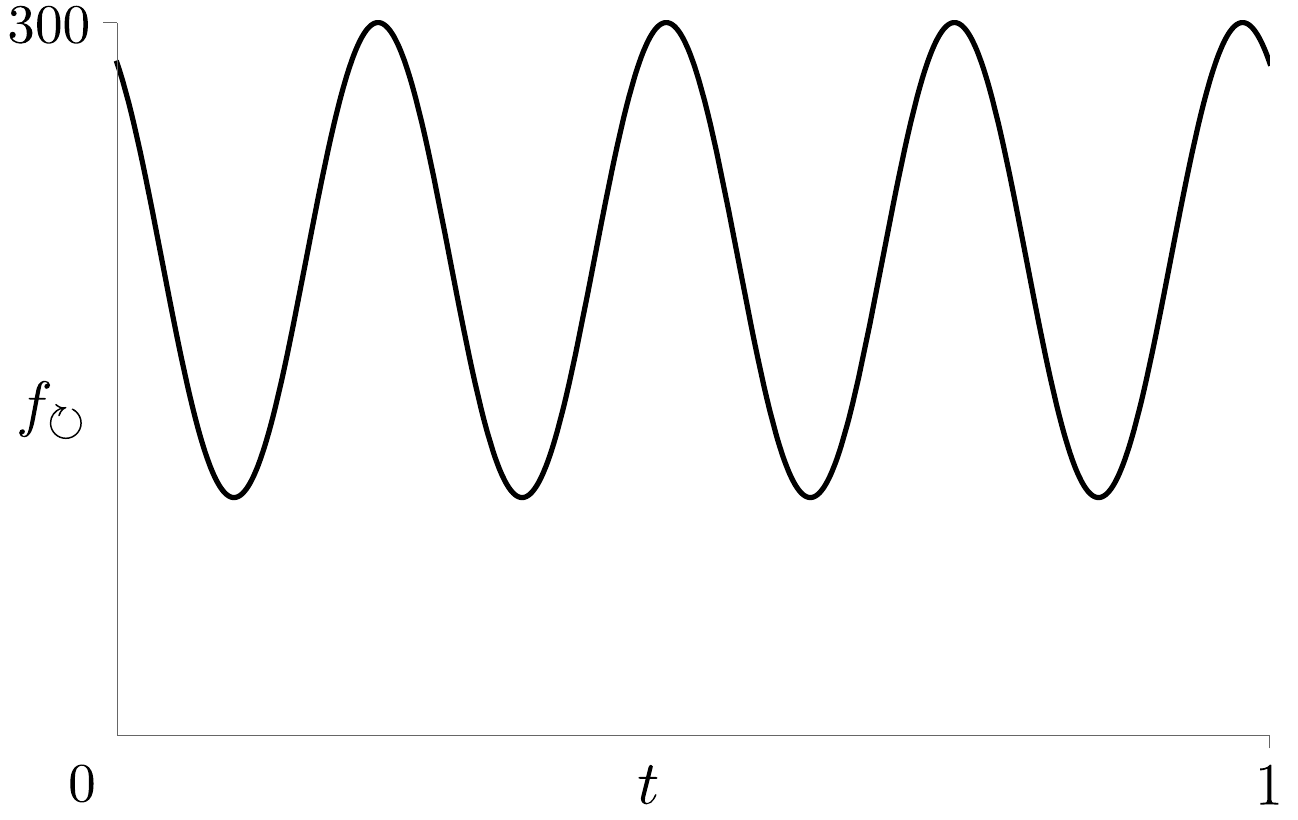}
	\caption{\small
		Applied force corresponding to expression~(\ref{eq:force}): $b=0.5$\,, $c=\pi/6$\,, $f_{0}=200\,{\rm N}$
	}
	\label{fig:FigRenzoTwo}
\end{figure}
For expression~(\ref{eq:speed}) to correspond to two revolutions per second, we modify it to be
\begin{equation}
	\label{eq:speed2}
	v_{\circlearrowright}(t)
	=
	4\pi r(1+a\cos(8\pi t))
	\,,\qquad
	t\in(0,1]
	\,,
\end{equation}
where $t$ stands for time; expression~(\ref{eq:speed2}) is illustrated in Figure~\ref{fig:FigRenzoOne}.
Accordingly, we modify expression~(\ref{eq:force}) to be
\begin{equation}
	\label{eq:force2}
	f_{\circlearrowright}(t)
	=
	f_{0}\,(1+b\cos(8\pi(t+c)))
	\,,\qquad
	t\in(0,1]
	\,,
\end{equation}
where $c$ is a time shift between $v_{\circlearrowright}$ and $f_{\circlearrowright}$\,, which is a constant whose units are seconds; expression~(\ref{eq:force2}) is illustrated in Figure~\ref{fig:FigRenzoTwo}.
Hence, expression~(\ref{eq:InsPow}) becomes
\begin{align}
	\nonumber
	\overline{P}_{\circlearrowright}
	&=
	\int\limits_{0}^{1}
	v_{\circlearrowright}(t)\,f_{\circlearrowright}(t)
	{\,\rm d}t
	=
	\int\limits_{0}^{1}
	\overbrace{
		4\pi r(1+a\cos(8\pi t))
	}^{v_{\circlearrowright}(t)}
	\overbrace{
		f_{0}(1+b\cos(8\pi(t+c)))
	}^{f_{\circlearrowright}(t)}
	{\,\rm d}t
	\\
	\label{eq:AveTime}
	&=
	2(2+ab\cos(8\pi c))\pi r f_{0}\,.
\end{align}
Examining expressions~(\ref{eq:InsPow}) and (\ref{eq:AveTime}), we see that to keep the same average power per second\,---\,with $ab=0$ or $c=0$\,---\,we need to halve the value of~$f_{0}$\,.
Otherwise, the ratio is
\begin{equation*}
	\dfrac{2+ab\cos(2\pi c)}{2\left(2 + ab\cos(8\pi c)\right)}
	\,.
\end{equation*}

In accordance with expression~(\ref{eq:InsPow}), which is tantamount to the power averaged over one second\,---\,if the cadence is one revolution a second\,---\,and given $a=0.25$\,, $b=0.5$\,, $c=0$\,, $f_{0}=200\,{\rm N}$\,, $r=0.175\,{\rm m}$\,, we have $\overline{P}_{\circlearrowright}=234\,{\rm W}$\,.
With a cadence of two revolutions a second, in accordance with expression~(\ref{eq:AveTime}), the same average power is obtained with~$f_{0}=100\,{\rm N}$\,.
Thus, among many possible pairs that result in $\overline{P}_{\circlearrowright}=234\,{\rm W}$\,, we have $(60$~rpm, $f_{0}=200\,{\rm N})$ and $(120$~rpm, $f_{0}=100\,{\rm N})$\,.

If $a=0$\,, in accordance with expression~(\ref{eq:InsPow}) and (\ref{eq:AvePow}), $\overline{P}_{\circlearrowright}=2\pi r f_{0}$\,, per second, and, in accordance with expression~(\ref{eq:AveTime}), $\overline{P}_{\circlearrowright}=4\pi r f_{0}$\,, per second.
Thus, to keep the same average power, we again halve the value of~$f_{0}$\,.
If the original value, at $60$~rpm, is $\overline{P}_{\circlearrowright}=234\,{\rm W}$\,, the corresponding value is calculated to be $f_{0}=213\,{\rm N}$\,, instead of~$200\,{\rm N}$\,.
This results in a different\,---\,and less accurate\,---\,pair, due to neglecting speed variation during a revolution.
\subsection{Closing remarks}
As illustrated in this appendix, there is a discrepancy between the power-meter calculations resulting from the use of the instantaneous-pedal-speed and average-pedal-speed information.
Removing this discrepancy is crucial for a variety of information that rely on power measurements, as is the case of this article.

Let us conclude by addressing the issue of sampling with regards to the discrepancy in power calculation resulting from $\overline{v}_{\circlearrowright}$ as opposed to $v_{\circlearrowright}(\theta)$\,. 
Let us consider the gear of $54\times17$\,.
For a road bicycle, one revolution results in a development of $6.67$~metres.
Hence, for the speed of $48.1$~kilometres per hour, a full rotation corresponds to half a second, which is a high time-trial cadence of $120$~revolutions per minute.
The sampling of twice-a-second, however, is insufficient for accurate information about power.
For that reason, the Favero power meters provide the cycling computer with data that already contains information based on the instantaneous\,---\,as opposed to the average\,---\,pedal speed.
\section{Time minimization with Lagrange multipliers}
\label{app:LagrangeMultipliers}
\setcounter{equation}{0}
\setcounter{figure}{0}
\renewcommand{\theequation}{\Alph{section}.\arabic{equation}}
\renewcommand{\thefigure}{\Alph{section}\arabic{figure}}
\subsection{Preliminary remarks}
Consider a flat course of length~$d$\,, whose one half is covered against the wind, as illustrated in Figure~\ref{fig:FigLagrange}, and the other half with the wind.
\begin{figure}
\centering
\includegraphics[scale=0.75]{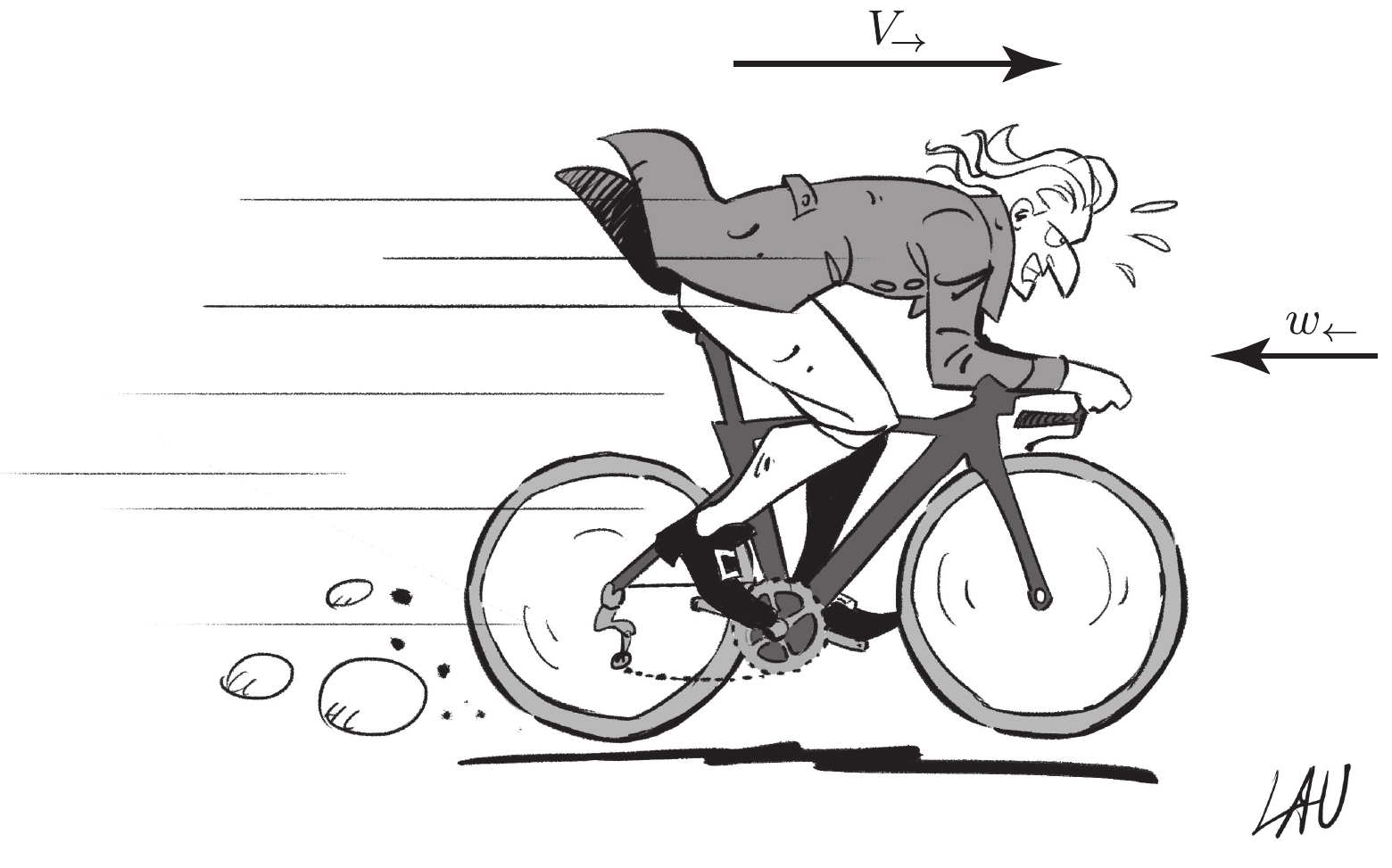}
\caption{\small Joseph-Louis (Giuseppe Luigi) Lagrange examining his optimizations based on the fact that, in general, in accordance with model~\eqref{eq:model}, headwinds,~$w_\leftarrow>0$\,, increase the air resistance, and tailwinds,~$w_\leftarrow<0$\,, decrease it, provided that $w_{\leftarrow}\nless-V_{\!\rightarrow}$\,.}
\label{fig:FigLagrange}
\end{figure}
To minimize the time,~$t$\,, we need to maximize the average speed,
\begin{equation}
	\label{eq:AveV}
	\overline V_{\!\rightarrow}
	=
	\dfrac{d}{t}
	=
	\dfrac{d}{\dfrac{d}{2\,V_{U}}+\dfrac{d}{2\,V_{D}}}
	=
	\dfrac{2\,V_{U}V_{D}}{V_{U}+V_{D}}
	\,,
\end{equation}
where $V_{U}$ and $V_{D}$ are the speeds on the upwind and downwind sections, respectively.
The maximum of this function occurs for all values along $V_{U}=V_{D}$\,.
To get a pair of values that corresponds to a realistic scenario, we invoke the method of Lagrange multipliers and find the maximum of speed~\eqref{eq:AveV}, subject to constraints.
To do so, we state the problem as a Lagrangian function of two variables with $n$ constraints, 
\begin{equation}
	\label{eq:L_general}
	L(V_{U},V_{D})
	=
	\overline V_{\!\!\rightarrow}(V_{U},V_{D})
	+
	\Lambda_{1}\,\Gamma_{1}(V_{U},V_{D})
	+
	\dots
	+
	\Lambda_{n}\,\Gamma_{n}(V_{U},V_{D})
	\,,	
\end{equation}
where $\Lambda_i$\,, with $i=1,\ldots,n$\,, is a Lagrange multiplier.
The optimization is achieved at the stationary points of function~\eqref{eq:L_general}, which we find by solving the system of equations,
\begin{equation}
	\label{eq:LagPar}
	\dfrac{\partial L}{\partial V_{U}}=0
	\,,\quad
	\dfrac{\partial L}{\partial V_{D}}=0
	\,,\quad
	\dfrac{\partial L}{\partial\Lambda_{1}}= 0
	\,,\quad\ldots\,,\quad
	\dfrac{\partial L}{\partial\Lambda_{n}}=0
	\,,
\end{equation}
whose solution is the pair,~$V_{U},V_{D}$\,, that extremizes expression~(\ref{eq:AveV}) and satisfies the constraints,~$\Gamma_i$\,, where $i=1,\ldots,n$\,, within the physical realm.
\subsection{Constraint of total work}
\label{app:ConstraintWork}
Let us impose a constraint in terms of the amount of total work, $W_{0}=W_{U}+W_{D}$\,, to be done by a cyclist on the upwind and downwind sections, whose proportions of length are stated in expression~(\ref{eq:AveV}),
\begin{equation}
	\label{eq:Gamma_{W}}
	\Gamma_{W}
	=
	\overbrace{
		\underbrace{
			\frac{
				{\rm C_{rr}}\,m\,g
				+
				\tfrac{1}{2}\,{\rm C_{d}A}\,\overline\rho
				\left(V_{U}+w_{\leftarrow}\right)^{2}
			}{
				1-\lambda
			}
		}_{F_{\!\leftarrow}}
		\dfrac{d}{2}
	}^{W_{U}}
	+
	\overbrace{
		\underbrace{
			\frac{
				{\rm C_{rr}}\,m\,g
				+
				\tfrac{1}{2}\,{\rm C_{d}A}\,\overline\rho
				\left(V_{D}-w_{\leftarrow}\right)^{2}
			}{
				1-\lambda
			}
		}_{F_{\!\leftarrow}}
		\dfrac{d}{2}
	}^{W_{D}}-W_{0}
	=
	0
	\,.
\end{equation}
Herein, we assume
\begin{equation*}
	W_{0}
	=
	\underbrace{
		\frac{
			{\rm C_{rr}}\,m\,g
			+
			\tfrac{1}{2}\,{\rm C_{d}A}\,\overline\rho\,{\overline V_{\!\rightarrow}}^{2}
		}{
			1-\lambda
		}
	}_{F_{\!\leftarrow}}
	d
\end{equation*}
to be the total amount of energy available to the cyclist, which corresponds to the work done on the same course, with a maximum effort\,---\,with no wind,~$w_{\leftarrow}=0\,{\rm m/s}$\,---\,resulting in a given value of $\overline V_{\!\rightarrow}$\,.
We write function~\eqref{eq:L_general} as
\begin{equation}
	\label{eq:L_{W}}
	L_{W}
	=
	\overline V_{\!\rightarrow}
	+
	\Lambda_{W}\,\Gamma_{W}
	\,.
\end{equation}
Considering $d=10000\,{\rm m}$\,, model parameters stated in Section~\ref{sec:FlatCourseVals}, namely, $m=111\,{\rm kg}$\,, $g=9.81\,{\rm m/s}^2$\,, $\overline\rho=1.204\,{\rm kg/m}^3$\,, ${\rm C_{d}A}=0.2607\,{\rm m}$\,, ${\rm C_{rr}}=0.002310$\,, $\lambda=0.03574$\,, and letting $\overline V_{\!\rightarrow}=10.51\,{\rm m/s}$\,, we obtain $W_{0}=2.059\times10^5\,{\rm J}$\,.
To minimize the traveltime with $w_{\leftarrow}=5\,{\rm m/s}$\,, we write system~(\ref{eq:LagPar}), in terms of function~(\ref{eq:L_{W}}),
\begin{equation}
	\label{eq:L_{W}_system}
	\left\{
	\begin{aligned}
		&\dfrac{\partial L_{W}}{\partial V_{U}}
		=
		\dfrac{2\,{V_{D}}^2}{\left(V_{D} + V_{U}\right)^{2}} 
		+ 
		\Lambda_{W}
		\left(1628\,V_{U}+8138\right)
		=
		0
		\,,
		\\
		&\dfrac{\partial L_{W}}{\partial V_{D}}
		=
		\dfrac{2\,{V_{U}}^2}{\left(V_{D} + V_{U}\right)^{2}}
		+
		\Lambda_{W}
		\left(1628\,V_{D}-8138\right)
		=
		0
		\,,
		\\
		&\dfrac{\partial L_{W}}{\partial\Lambda_{W}}
		= 
		813.8\left({V_{U}}^2 + {V_{D}}^2\right)
		+
		8138\left(V_{U}-V_{D}\right)
		-
		139100
		=
		0
		\,.
	\end{aligned}
	\right.
\end{equation}
Solving system~\eqref{eq:L_{W}_system} numerically, we obtain a single physical solution,
\begin{equation}
	\label{eq:L_{W}_system_sol}
	V_{U} = 8.279\,{\rm m/s}
	\quad{\rm and}\quad
	V_{D} = 11.68\,{\rm m/s}
\end{equation}
which is the pair that both maximizes expression~(\ref{eq:AveV}) and satisfies constraint~(\ref{eq:Gamma_{W}}).

In accordance with expression~(\ref{eq:AveV}), the average speed is $\overline V_{\!\rightarrow} = 9.689\,{\rm m/s}$\,, which is lower than the speed under the assumption of $w_{\leftarrow}=0\,{\rm m/s}$\,, namely, $\overline V_{\!\rightarrow}=10.51\,{\rm m/s}$\,.
This quantifies an adage that riding with the wind does not compensate for the speed lost by riding against the wind.
The loss is due to the dissipation of energy due to the air, rolling and drivetrain resistances, which are present on both the upwind and downwind sections.
\subsection{Constraint of average power}
\label{app:ConstraintPower}
Let us impose a constraint in terms of the value of average power, $P_0$\,, maintained by a cyclist on the upwind and downwind sections. 
In contrast to work, power is not a cumulative quantity.
Hence, the distance does not appear explicitly in a constraint, and we require constraints for both the upwind and downwind sections,
\begin{align}
	\label{eq:Gamma_{P_{U}}}
	\Gamma_{P_{U}}
	&=
	\overbrace{
		\underbrace{
			\frac{
				{\rm C_{rr}}\,m\,g
				+
				\tfrac{1}{2}\,{\rm C_{d}A}\,\overline\rho
				\left(V_{U}+w_{\leftarrow}\right)^{2}
			}{
				1-\lambda
			}
		}_{F_{\!\leftarrow}}
		V_{U}
	}^{P_{U}}
	-
	P_{0}
	\,,
	\\
	\label{eq:Gamma_{P_{D}}}
	\Gamma_{P_{D}}
	&=
	\overbrace{
		\underbrace{
			\frac{
				{\rm C_{rr}}\,m\,g
				+
				\tfrac{1}{2}\,{\rm C_{d}A}\,\overline\rho
				\left(V_{D}-w_{\leftarrow}\right)^{2}
			}{
				1-\lambda
			}
		}_{F_{\!\leftarrow}}
		V_{D}
	}^{P_{D}}
	-
	P_{0}
	\,.
\end{align}
Herein, we assume
\begin{equation*}
	P_{0}
	=
	\underbrace{
		\frac{
			{\rm C_{rr}}\,m\,g
			+
			\tfrac{1}{2}\,{\rm C_{d}A}\,\overline\rho\,{\overline V_{\rightarrow}}^{2}
		}{
			1-\lambda
		}
	}_{F_{\!\leftarrow}}
	\overline V_{\rightarrow}
\end{equation*}
to be the average power available to the cyclist, which corresponds to the average power achieved on the same course, with a maximum effort\,---\,with no wind,~$w_{\leftarrow}=0$\,---\,resulting in a given value of $\overline V_{\!\rightarrow}$\,.
Function~\eqref{eq:L_general} is
\begin{equation}
	\label{eq:L_{P}}
	L_{P}
	=
	\overline V_{\!\rightarrow}
	+
	\Lambda_{P_{U}}\,\Gamma_{P_{U}}
	+
	\Lambda_{P_{D}}\,\Gamma_{P_{D}}
	\,.
\end{equation}
Using the same model parameters in Appendix~\ref{app:ConstraintWork}, we obtain $P_{0} = 216.4\,{\rm W}$\,.
To minimize the traveltime with $w_{\leftarrow}=5$\,m/s\,, we write system~(\ref{eq:LagPar}), in terms of function~(\ref{eq:L_{P}}),
\begin{equation*}
	\left\{
	\begin{aligned}
		\dfrac{\partial L_{P}}{\partial V_{U}}
		&=
		\dfrac{2\,{V_{D}}^2}{\left(V_{D} + V_{U}\right)^{2}} 
		+ 
		\Lambda_{P_{U}}
		\left(0.4883\,{V_{U}}^2+3.255\,V_{U}+6.678\right)
		=
		0
		\,,
		\\
		\dfrac{\partial L_{P}}{\partial V_{D}}
		&=
		\dfrac{2\,{V_{U}}^2}{\left(V_{D} + V_{U}\right)^{2}}
		+
		\Lambda_{P_{D}}
		\left(0.4883\,{V_{D}}^2-3.255\,V_{D}+6.678\right)
		=
		0
		\,,
		\\
		\dfrac{\partial L_{P}}{\partial\Lambda_{P_{U}}}
		&=
		0.1628\,{V_{U}}^{3} 
		+ 
		1.628\,{V_{U}}^{2} 
		+ 
		6.678\,V_{U}
		-
		216.4
		=
		0
		\,,
		\\
		\dfrac{\partial L_{P}}{\partial\Lambda_{P_{D}}}
		&= 
		0.1628\,{V_{U}}^{3} 
		- 
		1.628\,{V_{U}}^{2} 
		+ 
		6.678\,V_{U} 
		-
		216.4
		=
		0
		\,.
	\end{aligned}
	\right.
\end{equation*}
The single physical solution is
\begin{equation}
	\label{eq:L_{P}_system_sol}
	V_{U} = 7.603\,{\rm m/s}
	\quad{\rm and}\quad
	V_{D} = 13.92\,{\rm m/s}
	\,,
\end{equation}
which both maximizes expression~\eqref{eq:AveV} and satisfies constraints~\eqref{eq:Gamma_{P_{U}}} and~\eqref{eq:Gamma_{P_{D}}}.
The corresponding average is $\overline V_{\!\!\rightarrow} = 9.833\,{\rm m/s}$\,, which confirms an adage that riding with the wind does not compensate for the speed lost by riding against the wind.
\subsection{Relation between differences and derivatives}
To conclude this appendix, let us comment on difference~$\Delta V_{\!\rightarrow}/\Delta w_{\leftarrow}$\,, discussed herein, in the context of $\partial_{w_{\leftarrow}}V_{\!\rightarrow}$\,, whose value is presented in Table~\ref{table:PhysicalRates}.
Partial derivatives correspond to a tangent to a curve at a point, and the differences to a secant over a segment of the curve.
Also, partial derivatives are obtained under the assumption that all other quantities are constant.

The latter requirement is satisfied in Appendix~\ref{app:ConstraintPower}, where
\begin{equation*}
	\frac{\Delta V_{\rightarrow}}{\Delta w_{\leftarrow}}
	=
	\frac{V_{U}-V_{D}}{w_{\leftarrow}-(-w_{\leftarrow})}
	=
	\frac{7.603\,{\rm m/s}-13.92\,{\rm m/s}}{10\,{\rm m/s}}
	=
	-0.6314\,,
\end{equation*}
which agrees with $\partial_{w_{\leftarrow}}V_{\rightarrow}$\,, in Table~\ref{table:PhysicalRates}, to two decimal points.
For $w_{\leftarrow}=0.05\,{\rm m/s}$\,, we obtain $V_{U} = 10.48\,{\rm m/s}$ and $V_{D} = 10.54\,{\rm m/s}$\,; hence, $\Delta V_{\rightarrow}/\Delta w_{\leftarrow}=-0.6359$\,, which agrees with $\partial_{w_{\leftarrow}}V_{\rightarrow}$ to six decimal points.
In general,
\begin{equation*}
	\lim_{\Delta w_{\leftarrow}\rightarrow0}
	\frac{\Delta V_{\rightarrow}}{\Delta w_{\leftarrow}}
	=
	\frac{\partial V_{\rightarrow}}{\partial w_{\leftarrow}}
	\,,
\end{equation*}
as expected, in view of a secant approaching a tangent.

The requirement of constant quantities is not satisfied in Appendix~\ref{app:ConstraintWork}, since $P$ is allowed to vary to maintain the imposed value of~$W_{0}$\,.
In Appendix~\ref{app:ConstraintPower}, $W$ varies to maintain the imposed value of~$P_0$\,, but $W$ is not a variable in function~(\ref{eq:f}), used in partial derivatives.

As shown in this appendix, properties of partial derivatives need to be considered in examining time-trial strategies.
In contrast to common optimization methods, partial derivatives correspond to a change of a single variable, only.
\subsection{Closing remarks}
Let us examine the constraints discussed in this appendix in terms of required powers.
For the work constraint, the average speed is~$\overline V_{\!\rightarrow}= 9.689\,{\rm m/s}$\,.
Following model~\eqref{eq:model}, the required powers are $P_{U}=365.5\,{\rm W}$ and $P_{D}=59.95\,{\rm W}$\,, for the upwind and downwind sections, respectively.
For windless conditions, we have $\overline{P} =173.3\,{\rm W}$\,.
Thus, $P_{U}$ is significantly greater than~$\overline P$\,.

For the power constraint, with $P_{0}=216.4\,{\rm W}$\,, the average speed is~$\overline V_{\!\rightarrow}= 9.833\,{\rm m/s}$\,.
Since the average speed is greater than for the work-constraint optimization and the average power does not exceed the value obtained in windless conditions, this appears to be the preferable strategy.
Also, power is provided as an instantaneous quantity by the power meters, which allows the cyclist to follow a given strategy, whose further refinements are to be considered in future studies.

To close, let us consider a rule of thumb of~\citet{Anton2013}.
\begin{quote}
	Choose a target-speed~$v_0$\,.
	$\,[\ldots]\,$
	Endeavor to ride at $v\cong v_0+w/4$ when the wind is at your back and at
	$v\cong v_0-w/2$ when the wind is at your face.	
\end{quote}
If we choose $\overline V_{\!\rightarrow}=10.51\,{\rm m/s}=:v_0$ to be a target speed, with $w=5\,{\rm m/s}$\,, speeds~\eqref{eq:L_{P}_system_sol}, which result from the power constraint, are less congruent with this rule than speeds~\eqref{eq:L_{W}_system_sol}, which result from the work constraint, yet\,---\,according to the present analysis\,---\,speeds~\eqref{eq:L_{P}_system_sol} appear to be preferable.
This is an indication of further subtleties that need to be considered in developing a time-trial strategy.
\end{appendix}
\end{document}